\documentclass[journal]{IEEEtran}
\usepackage{amsmath,amssymb,amstext,graphicx,cite}

\newtheorem{Def}{\bf Definition}
\newtheorem{Rmk}{\bf Remark}
\newtheorem{Exam}{\bf Example}
\newtheorem{Assn}{\bf Assumption}
\newtheorem{Thm}{\bf Theorem}
\newtheorem{Lem}{\bf Lemma}
\newtheorem{Cor}{\bf Corollary}
\newcommand{\HH}{\mathcal{H}}
\newcommand{\MM}{\mathcal{M}}
\newcommand{\CC}{\mathcal{C}}
\newcommand{\C}{\rm c}
\newcommand{\cc}{_{\rm c}}
\newcommand{\rr}{_{\rm r}}
\newcommand{\pp}{_{\rm p}}

\newcommand{\DD}{\mathcal{D}}
\newcommand{\EE} {\mathcal{E}}

\newcommand{\R}{\mathbb{R}}

\newcommand{\End}{\square}
\newcommand{\GG}{\mathcal{G}}
\newcommand{\bs}{\backslash}
\newcommand{\p}{\partial}
\newcommand{\dd}{{\rm d}}

\newcommand{\KK}{\mathcal{K}}
\newcommand{\VV}{\mathcal{V}}
\newcommand{\WW}{\mathcal{W}}

\newcommand{\OO}{\mathcal{O}}
\newcommand{\BB}{\mathcal{B}}

\newcommand{\RR}{\mathcal{R}}

\newenvironment{proof}{{\em Proof.}}{\hfill $\blacksquare$}

\begin{document}

\title{State Estimation and Tracking Control for Hybrid Systems by Gluing the Domains}
\author{Jisu~Kim,~Hyungbo~Shim,~and Jin Heon Seo%
\thanks{Preliminary versions of this paper have been presented in \cite{JS14} and \cite{JS16}.}%
\thanks{This research was supported in part by a grant to Bio-Mimetic Robot Research Center funded by Defense Acquisition Program Administration, and by Agency for Defense Development (UD160027ID), and in part by the National Research Foundation of Korea (NRF) grant funded by the Korea government (Ministry of Science and ICT) (No. NRF-2017R1E1A1A03070342).}%
\thanks{J. Kim, H. Shim, and J. H. Seo are with ASRI, Dept.~of Electrical and Computer Engineering, Seoul National University, Korea.
{\tt\small \{wltnqhdn, hshim, jhseo\}@snu.ac.kr}}%
}

\maketitle

\begin{abstract}
We study the design problems of state observers and tracking controllers for a class of hybrid systems whose state jumps. 
The idea is to utilize the well-known method of gluing the jump set (a part of domain where the jumps take place) onto its image, which converts the hybrid system into a continuous-time system whose state does not jump. 
Sufficient conditions for this idea to be implemented are listed and discussed with a few concrete examples.
In particular, we present a structural condition for an observer design, and, for tracking control, we introduce a feedback to compensate residual discontinuity in the vector field after gluing.
The benefits of the proposed approach include that the observer design does not require detection of the state jumps, and that the tracking control does not require the plant state jumps when the reference jumps.
\end{abstract}


\section{Introduction}

A dynamical system that exhibits both continuous and discrete dynamic behaviors is called a hybrid system.
The modeling frameworks that involve the continuous-time dynamics, the discrete events, and the interaction between them have been studied in, e.g., \cite{JL03,GS09}. 
In these frameworks, the state estimation and the control problems have been studied for special classes of hybrid systems such as switched systems \cite{SG12,BM12}, mechanical systems with impacts \cite{BB00,LM011,AT13,Men16}, and polyhedral billiards with impacts \cite{FT13}. 

However, many existing observer designs require that the jump times of the observer state should coincide with those of the plant state (see, e.g.,\cite{LM011,AT13,Za14}).
Similarly, in many existing reference tracking controllers, the jump times of the plant state and the reference trajectory should coincide (see, e.g.,\cite{BB00,RG14,MR11}). 
This is because, in particular, if two states of the plant and the estimator/reference do not jump at the same time when they are close to each other, then the estimation/tracking error may become abruptly large immediately after one of them jumps far from the other while the other is yet to jump.
When this occurs, it is not straightforward for the estimator/controller to determine its action.
In order to synchronize their jump times, the state estimator needs the jump time information of the plant, and the tracking controller needs to have a method to enforce the plant state to jump whenever the reference jumps.
Unfortunately, both of them are not easy in some practices.

An idea to overcome this difficulty is to re-think the distance between two states.
One way to do this is to allow comparison of two states at different times; for example, by taking the minimal distance between two trajectories ``around'' one's jump time, unnecessarily large error can be avoided.
See \cite{BA02,GS09} for this line of thought.
Another way is to take the jump map into account when defining the distance; that is, if two states are related by the jump map, their distance is considered to be zero.
Related results are found in \cite{FT13,BBJ13}.
Similar, but more radical idea is to deform the underlying domain of the hybrid system in the state-space, so that two different states are glued into one if they are related by the jump map.
To the best of the authors' knowledge, application of the gluing idea to the hybrid dynamic systems was firstly studied in \cite{SS05}, where the notion of `hybrifold' is introduced.
This idea is then utilized for state estimation problem in \cite{AS09}.
In this paper, we continue the study of \cite{AS09} by applying the idea for ``designing'' state observers and tracking controllers.
In particular, the jump set---where the jumps take place---is glued onto its image---the destination of the jump---through the jump map.
Then, because the system becomes a continuous-time system without any jumps of the state, the conventional observer or controller design techniques may be applied.
For this, we refine the notion of ``gluing function'' in order to be used for estimation and tracking control problems (Section \ref{Section: Preliminaries}).
We also discuss the need for matching vector fields before and after the jump.
Then, we present a structural condition for a class of hybrid systems that guarantee the matching condition in the case of observer design problem (Section \ref{Section III}), and a method to satisfy the matching condition {\em by feedback} in the case of tracking control problem (Section \ref{Section IV}).
As a result, jump time detection is not necessary for observer design, and the tracking controller does not need to make the plant state jump whenever the reference jumps. 
Three concrete examples demonstrate effectiveness of the proposed approaches.

{\em Notation:}
The zero vector in $\R^n$ is denoted by $0_n$. 
Given $x\in \R^n$, $|x|$ denotes the Euclidean vector norm and $d_\CC(x)$ means the distance of $x$ to a subset $\CC\subset\R^n$, i.e., $d_\CC(x) := \inf_{y\in \CC}|x-y|$. 
An open ball at $x\in\R^n$ with a radius $\epsilon>0$ is denoted by $\BB_x(\epsilon)$. 
For given column vectors $x$ and $y$, we define $(x,y) := [x^\top~y^\top]^\top$.
For a subset $\CC\subset \R^n$ and a $C^1$ function $\psi:\CC\rightarrow\R^m$, $\dd\psi$ denotes the Jacobian matrix of $\tilde \psi$ where $\tilde \psi$ is an extension of $\psi$ to some open neighborhood of $\CC$ and is continuously differentiable. 
For a given differentiable function $r:\R^n\rightarrow \R$, $\nabla r$ is the gradient of $r$, i.e., $\nabla r:=[\p r/\p x_1, \cdots, \p r/\p x_n]$. 
The symbol $\varnothing$ implies the null set.

\section{Gluing the Domain}\label{Section: Preliminaries}

\subsection{Standing Assumptions}

We first introduce our framework for hybrid systems. 
The following definition is motivated by \cite{GS09,BBJ13,JL03}. 
A {\em hybrid system} is a 4-tuple $\HH=(\CC,f,\DD,g)$ represented by
\begin{align*}
\HH: \left\{
\begin{array}{lll}
    \dot x &\!\!\!\!= f(x,u)   &\quad \hbox{when $x\in\CC \subset \R^n$} \\
     x^{+} &\!\!\!\!= g(x)   &\quad \hbox{when $x\in\DD \subset \R^n$}
\end{array}
\right.
\end{align*}
with $x \in \R^n$ and $u \in \R^p$ being the state and the input, respectively.
The differential equation, governing the continuous-time evolution of the state $x$ when $x$ belongs to the {\it flow set} $\CC \subset\R^n$, and the difference equation, determining the discrete update of $x$ when $x$ belongs to the {\it jump set} $\DD \subset \R^n$, are given by the {\it flow map} $f:\CC \times \R^p \to \R^n$ and the {\it jump map} $g: \DD \to \R^n$, respectively. 

We now recall several definitions regarding the solutions to the hybrid systems, adopted from \cite{JL03}.
A {\em hybrid time trajectory} is a finite or infinite sequence of intervals $\tau=\{I_i\}_{i=0}^N$ ($N$ can be $\infty$) such that
\begin{itemize}
\item $I_i=[\tau_i,\tau_i']$ with $\tau_i\leq\tau_i'=\tau_{i+1}$ for all $0\leq i<N$; in particular, $\tau_0=0$;
\item when $N<\infty$, either $I_N=[\tau_N,\tau_N']$ or $I_N=[\tau_N,\tau_N')$.
\end{itemize}
For $\tau=\{I_i\}_{i=0}^N$, let $\langle\tau\rangle := \{0,1,\ldots,N\}$ (possibly $N=\infty$), and $|\tau| := \sum_{i\in\langle\tau\rangle}(\tau_i'-\tau_i)$.
We say that $\tau=\{I_i\}_{i=0}^N$ is a {\em prefix} of $\tilde\tau=\{{\tilde I_i}\}_{i=0}^{\tilde N}$ and write $\tau \sqsubseteq \tilde\tau$, either if they are identical or if $N$ is finite, $N \leq \tilde N$, $I_i = \tilde I_i$ for all $0 \leq i < N$, and $I_N \subseteq \tilde I_N$. 
With $u=u(t,x)$, an {\em execution of $\HH$} excited by an initial condition $x_0\in\CC\cup \DD$ is a pair $\chi=(\tau,\xi)$ where $\tau$ is a hybrid time trajectory and $\xi = \{\xi^i : i \in \langle \tau \rangle\}$ is a collection of absolutely continuous maps $\xi^i : I_i \rightarrow \CC$ such that
\begin{itemize}
\item $\xi^0(0) = x_0$;
\item $\dot \xi^i(t) = f(\xi^i(t),u(t,\xi^i(t)))$ for all $i\in\langle\tau\rangle$ and for almost all $t\in(\tau_i,\tau_i')$;
\item $g(\xi^i(\tau_i')) = \xi^{i+1}(\tau_{i+1})$ for all $i\in\langle\tau\rangle\bs\{N\}$.
\end{itemize}
We say that an execution $\chi=(\tau,\xi)$ of $\HH$ is a {\em prefix} of another execution $\tilde\chi=(\tilde\tau,\tilde \xi)$ of $\HH$ and write $\chi\sqsubseteq\tilde\chi$, if $\tau\sqsubseteq\tilde\tau$ and $\xi^i(t)=\tilde \xi^i(t)$ for all $i\in\langle\tau\rangle$ and $t\in I_i$. 
We say that $\chi$ is a {\em strict prefix} of $\tilde\chi$, if $\chi\sqsubseteq\tilde\chi$ and $\chi\neq\tilde\chi$. 
An execution is called {\em maximal} if it is not a strict prefix of any other executions. 
An execution is called {\em infinite} if either $N=\infty$ or $|\tau|$ is not finite. 
Otherwise, it is called {\em finite}. 
We say an execution is {\em infinite in $t$-direction} when $|\tau|$ is not finite.
For each maximal execution, a ``{\em state trajectory}'' $x:[0,|\tau|)\rightarrow\CC$ is given by 
$$x(t) := \xi^{i(t)}(t) \qquad \text{for each $t\in[0,|\tau|)$}$$
where $i:[0,|\tau|)\to\langle\tau\rangle$ is the function satisfying $t\in[\tau_{i(t)},\tau_{i(t)}')$. 
Note that $i(t)$ is uniquely defined for each $t\in[0,|\tau|)$.	 

In this paper, we consider a plant $\HH_{\rm p}$ satisfying the following assumptions. 

\begin{Assn}\label{Ass:standing ass}
The plant $\HH_{\rm p}=(\CC,f,\DD,g)$ satisfies that\\
\noindent{\bf (A1)} $\CC\subset\R^n$ is a $k$-dimensional smooth submanifold with boundary\footnote{
We follow the convention of the terminology in \cite{JML03}. In this regard, every point of $\CC$ has an open neighborhood in $\CC$ diffeomorphic to an open subset of $\R^k$ or of a half-space $\R^{k-1} \times [0,\infty)$, and the boundary $\partial \CC$ is the set of points satisfying only the latter (i.e., points whose neighborhoods are diffeomorphic to only the half-space).}; if $k<n$, there exists a smooth map $r_\CC:\R^n\to \R^{n-k}$ such that $r_\CC(x)=0_{n-k}$ and $\dd r_\CC(x)$ has full row rank for all $x\in \CC$; \\
\noindent{\bf (A2)} $f$ is smooth; if $k<n$, then $\dd r_\CC (x) f(x,u)=0$ for all $(x,u)\in \CC\times \R^p$; \\
\noindent{\bf (A3)} $\DD$ consists of connected components of the boundary $\partial \CC$ of $\CC$; \\	
\noindent{\bf (A4)} $g$ can be extended to a diffeomorphism $\tilde g :\tilde \DD \to \tilde g(\tilde \DD)$ where $\tilde \DD$ is an open neighborhood of $\DD$ in $\R^n$.
Furthermore, $\GG := g(\DD)$ consists of connected components of $\partial \CC$ such that $\DD \cap \GG =\varnothing$ and $\DD\cup\GG=\partial \CC$. 
Finally, there exist smooth maps $r_\DD:\R^n\to\R$ and $r_\GG:\R^n\to\R$ satisfying that $\DD=\{x\in \CC : r_{\DD}(x)=0\}$, $\GG=\{x\in \CC : r_{\GG}(x)=0\}$, $\CC \subset \{x \in \R^n : (r_{\DD}(x) \leq 0) \;\&\; (r_{\GG}(x) \geq 0) \}$ where $\dd r_\DD(x)$ and $\dd r_\GG(x)$ have full row rank for all $x$ in $\DD$ and $\GG$, respectively.	\hfill $\End$	
\end{Assn}

Under (A4), the state jumps occur from some part of the boundary to another part of the boundary. Moreover, it does not allow the consecutive jumps because $\DD \cap\GG= \varnothing$.

\subsection{Gluing Function}

Under Assumption \ref{Ass:standing ass}, the state trajectory $x(t)$ of $\HH\pp$ excited at $x_0$ is right-continuous in $t$, but it is not left-continuous. 
In order to deal with the discontinuities, the following gluing function is introduced.

\begin{Def}\label{Definition: Gluing Function}
A function $\psi:\CC\rightarrow\R^m$ $(m \geq k)$ is called a {\em gluing function} of the system $\HH_{\rm p}=(\CC,f,\DD,g)$ if it satisfies \\
\noindent{\bf (G1)} $\psi(x)=\psi(g(x))$ for all $x\in\DD$;\\
\noindent{\bf (G2)} $\psi|_{\CC\bs\DD}$, which is $\psi$ restricted on the domain $\CC\bs\DD$ (set subtraction), is injective;\\
\noindent{\bf (G3)} $\psi$ is of class $C^1$ (continuously differentiable);\\
\noindent{\bf (G4)} if $k=n$, ${\rm rank} (\dd\psi(x))=n$ for all $x\in\CC$; if $k<n$, $\dd\psi(x)v\neq 0_m$ for each $x\in\CC$ and each non-zero vector $v\in \ker(\dd r_\CC(x))$;\\ 
\noindent{\bf (G5)} for every compact set $K\subset \psi(\CC)$, the preimage $\psi^{-1}(K)$ is compact.\\
We call $\psi(\CC)$ a {\em glued domain} and denote it by $\CC^\psi$.   \hfill $\End$
\end{Def}

\begin{Rmk}\label{Rem:gluing function}
By (G1), the gluing function does glue the boundary $\DD$ of $\CC$ to the boundary $\GG$ of $\CC$. 
Furthermore, by (G2), $\psi$ is a natural projection. 
We impose (G4) to guarantee that, for each $x\in\CC$, the map $v \mapsto \dd \psi(x)v$ from the tangent space at $x$ is injective. 
Finally, (G5) implies that, if a sequence $\{x_i\}$ escapes to infinity in $\CC$, i.e., if no compact set in $\CC$ contains infinitely many terms of $\{x_i\}$, then the sequence $\{\psi(x_i)\}$ also escapes to infinity in $\psi(\CC)$.
\hfill $\End$	
\end{Rmk}

\begin{Rmk}\label{rem:2}
Under Assumption \ref{Ass:standing ass} for $\HH\pp$, there always exists a gluing function \cite[Theorem 9.29]{JML03}. 
In fact, the gluing function of Definition \ref{Definition: Gluing Function} is a proper $C^1$ embedding of a smooth manifold $\CC/\!\sim$ which is the quotient space of $\CC$ determined by the relation $x \sim g(x)$. 
Note that since a smooth structure on $\CC/\!\sim$ exists, we can regard $\CC/\!\sim$ as a smooth $k$-manifold. 
Since the natural projection $\pi$ sending each point of $\CC$ to the equivalence class is proper under Assumption~\ref{Ass:standing ass} and, by Whitney Embedding Theorem, the smooth $k$-manifold $\CC/\!\sim$ admits a proper smooth embedding $\iota$ into $\R^{2k+1}$, we can take a gluing function as $\psi:=\iota\circ\pi$. 
However, constructing a gluing function in this way is a daunting task  and the actual gluing function is often found by intuition and trial-and-error.
See the examples in Sections \ref{Section III} and \ref{Section IV}. \hfill $\End$
\end{Rmk}

\begin{Lem}\label{Lem: local embedding of GF}
Suppose there exists a gluing function $\psi$ for $\HH\pp$. 
Then, for any compact set $\MM \subset \CC$ satisfying that $\MM\cap \DD=\varnothing$ or $\MM\cap \GG=\varnothing$, it holds that $\psi$ is injective on $\MM$ and there exists $L>0$ such that, for all $(x,y)\in \MM \times \MM$,
\begin{align*}
|x-y|\leq L|\psi(x)-\psi(y)|.
\end{align*}
\end{Lem}

\begin{proof}
	See the Appendix.
\end{proof}

When there is a gluing function $\psi$ of $\HH\pp$, the state trajectory $x(t)$ can be expressed on the glued domain $\CC^\psi$ as
\begin{align*}
\zeta(t) := \psi(x(t)) \quad \in\CC^\psi.
\end{align*}
Since all the discontinuities of $x(t)$ are ``glued" by (G1), $\zeta(t)$ is continuous with respect to $t$. In this sense, we call $\zeta(t)$ the {\it glued trajectory} of $x(t)$ by $\psi$.

We now define the {\it inverse gluing function} induced from $\psi$. 
As a matter of fact, $\psi^{-1}$ on $\psi(\CC)$ cannot be defined due to the loss of injectivity by (G1). 
However, noting that $\psi|_{\CC\bs\DD}$ is bijective onto its image $\psi(\CC\bs\DD)$ by (G2) and that $\CC^\psi = \psi(\CC) = \psi(\CC\bs\DD)\cup\psi(\DD) = \psi(\CC\bs\DD)$ because $\psi(\DD) = \psi(\GG)$ by (G1) and $\psi(\GG)\subset\psi(\CC\bs\DD)$ by $\DD\cap\GG = \varnothing$, we can instead define the inverse gluing function from $\CC^\psi$ to $\CC$ by (we abuse notation by writing it as $\psi^{-1}$)
\begin{align} \label{eq: Inverse Function of Gluing Function}
\psi^{-1}(\zeta) := \psi|_{\CC\bs\DD}^{-1}(\zeta) \quad \hbox{for all $\zeta\in\CC^\psi$}.
\end{align}

\section{State Estimation by Gluing}\label{Section III}

In this section, we propose a state estimation strategy using the gluing function, which does not require any detection of the time instants when the state jumps. 

\subsection{Standing Assumptions}

To deal with the state estimation problem, we define the plant $\HH\pp=(\CC,f,\DD,g)$ with an output $y=h(x)$ as $\HH^h\pp:=(\CC,f,\DD,g,h)$ where $h:\CC\to \R^q$ is an {\it output function}, and assume that $f$ is independent of $u$. 
We impose the following additional assumptions on $\HH\pp^h$ under consideration. 

\begin{Assn}\label{Ass:basic estimation}
$\HH\pp^h=(\CC,f,\DD,g,h)$ satisfies that \\ 
\noindent{\bf (E1)} there is a compact subset $\EE\subset \CC$ such that
$$x(0) \in \EE \quad \Rightarrow \quad x(t) \in \EE \text{ for all } t \geq 0;$$
\noindent{\bf (E2)} $\left\{ \begin{array}{ll} \nabla r_\DD (x)\cdot f(x)>0& \hbox{for all } x\in\DD, \\
\nabla r_\GG (x)\cdot f(x)>0& \hbox{for all } x\in\GG; \end{array}\right.$\\
\noindent{\bf (E3)} $h$ is continuous and
\begin{align}\label{eq:output matching condition}
h(x) = h(g(x)) \text{ for all } x\in\DD.
\end{align}
\hfill $\End$
\end{Assn}

From (A1--4) and (E1--2), a state trajectory $x(t)$ of $\HH\pp^h$ starting in $\EE$ is uniquely well-defined for all $t \ge 0$. 
And, by (E3), there is no jump in the output value when the state jump occurs, which actually makes it difficult to detect the state jumps by merely observing the output.
We call \eqref{eq:output matching condition} the {\it output matching condition}.
Finally, (E2) implies that the states cannot stay in $\DD$ and $\GG$ during the continuous-time evolution, and in fact, they cannot stay ``near'' $\DD$ and $\GG$ as seen in the following lemma.

\begin{Lem} \label{Lem: Existence of alpha and beta}
Let $\epsilon$-neighborhoods of $\GG$ and $\DD$ in $\EE$ be defined as $\OO_\GG(\epsilon) := \{x\in\EE:d_{\GG\cap \EE}(x)<\epsilon\}$ and $\OO_{\DD}(\epsilon) := \{x\in\EE : d_{\DD\cap \EE}(x) < \epsilon \}$, respectively.
Then, for a plant $\HH\pp^h$ satisfying (A1--4) and (E1--2), there exists a class-$\KK$ function $\alpha$ such that the solution $x(t)$ starting in $\EE$ belongs to $\OO_\GG(\epsilon)$ or $\OO_\DD(\epsilon)$ only if $\tau_i - \alpha(\epsilon) < t < \tau_i + \alpha(\epsilon)$ for some $i$ where $\tau_i$ is a jump instant corresponding to $x(t)$; in other words,
$$x(t) \notin \OO_\GG(\epsilon)\cup\OO_\DD(\epsilon), \;
\forall t \in \tau_\alpha(\epsilon) := \R_{\geq0} \bs \cup_{i=0}^{N} \BB_{\tau_i}(\alpha(\epsilon)).$$
\end{Lem}

\begin{proof}
	See the Appendix.
\end{proof}

\subsection{State Estimation by Gluing}

Under Assumption \ref{Ass:standing ass}, a gluing function $\psi$ exists for $\HH\pp^h$ (see Remark \ref{rem:2}).
Then, the vector field (i.e., the flow map) at $x$ of $\CC$ is naturally mapped to the vector field at $\zeta:=\psi(x)$ of $\CC^\psi$ by the linear map $\dd \psi(x)$. 
Due to (G1) and (G2), if $\zeta\in \CC^\psi\bs\psi(\DD)$, there exists a unique $x\in\CC$ such that $\psi(x) = \zeta$.
Therefore, we obtain $\dd \psi(x) f(x)$ as the vector field at $\zeta \in \CC^\psi \bs \psi(\DD)$.
However, if $\zeta \in \psi(\DD)$, there exist two points $x \in \DD$ and $x' := g(x) \in \GG$ such that $\zeta=\psi(x)=\psi(x')$ and we have $\dd \psi(x) f(x)$ and $\dd \psi(g(x)) f(g(x))$ as the candidates of the vector field at $\zeta$. 
This implies that the vector fields may not be glued at a point belonging to $\psi(\DD)=\psi(\GG)$ even if we glue the domain by $\psi$. 
This observation motivates us to propose the {\it vector field matching condition}: 
\begin{align}\label{eq: matching condition of estimation}
\dd\psi(x)f(x)=\dd \psi(g(x)) f(g(x)) \hbox{ for all $x\in \DD$}
\end{align}
which implies that the vector fields are glued on $\psi(\DD)$ by $\psi$. 
By the smoothing theorems in \cite[Theorem 8.2.1]{Hir76} or \cite[Lemma 3]{Sas15}, there exists a gluing function that satisfies \eqref{eq: matching condition of estimation} as well.

Since the vector field at $\zeta\in\psi(\DD)$ is now uniquely defined by \eqref{eq: matching condition of estimation}, the vector field at $\zeta\in\CC^\psi$ can be represented as $\dd \psi(x) f(x)$ where $\psi(x)=\zeta$. 
Accordingly, for all $\zeta \in \CC^\psi$, the corresponding vector field is obtained as 
\begin{align}\label{eq: flow after gluing}
f^\psi(\zeta) :=\dd \psi(\psi^{-1}(\zeta)) f(\psi^{-1}(\zeta))
\end{align} 
which is a function of $\zeta$. 
It is seen that $f^\psi$ is continuous on $\CC^\psi$.
Indeed, both $\dd\psi$ and $f$ are continuous on $\CC$ by (G3) and (A2), and $\psi^{-1}$ is continuous on $\CC^\psi \bs \psi(\DD)$ by Lemma \ref{Lem: local embedding of GF}.
Therefore, with \eqref{eq: matching condition of estimation}, the continuity of $f^\psi$ on $\CC^\psi$ follows.
Similarly, we obtain the output function on $\CC^\psi$ as 
\begin{align}\label{eq:output after gluing}
h^\psi(\zeta):=h(\psi^{-1}(\zeta)), \qquad \zeta\in\CC^\psi,
\end{align}
which is also continuous because of (E3).
Then, $\psi$ converts $\HH^h\pp$ into the following continuous-time system:
\begin{subequations}\label{System: on Glued Domain}
\begin{align}
\dot\zeta&=f^\psi(\zeta),  \qquad \zeta\in\CC^\psi\subset\R^m, \label{System: Dynamis on Glued Domain}\\
y&=h^\psi(\zeta) \label{System: Output on Glued Domain}.        
\end{align}
\end{subequations}
It trivially follows that $\zeta(t)=\psi(x(t))$ is the solution to \eqref{System: Dynamis on Glued Domain} starting from $\psi(x_0)\in\CC^\psi$. 
Furthermore, since $\EE$ is forward invariant set of $\HH\pp^h$, $x(t)$ is defined on $\CC$ for all $t\geq0$, which implies that $\psi(\EE)$ is forward invariant under $f^\psi$.

The system \eqref{System: on Glued Domain} is a continuous-time system without any discrete events and we call it a {\it glued system} of $\HH^h\pp$ by $\psi$. 
If the glued system is of the form that admits a conventional observer design method for continuous-time systems, then we can design an observer and obtain an estimate $\hat\zeta(t)$ for $\zeta(t)=\psi(x(t))$. 
If $\hat \zeta(t)$ belongs to $\CC^\psi$, then, through the inverse gluing function $\psi^{-1}:\CC^\psi\to \CC$ in \eqref{eq: Inverse Function of Gluing Function}, the estimate for $x(t)$ is obtained.
However, the estimate $\hat\zeta(t)$, which is defined on $\R^m$, may go out of the domain $\CC^\psi$ in general, for example, during the transient.
To deal with this issue, a projection map $\Pi_{ \psi(\EE)}:\R^m \to \psi(\EE)$ is introduced, which satisfies that
\begin{align*}
\Pi_{\psi(\EE)}(\hat\zeta) \in \{ \zeta' : \arg \min_{\zeta'\in \psi(\EE)}|\hat\zeta-\zeta'| \}. 
\end{align*} 
With $\Pi_{\psi(\EE)}$, the estimate of $x(t)$ is defined by
\begin{align*}
\hat x(t):=\psi^{-1}(\Pi_{\psi(\EE)}(\hat{\zeta}(t))).
\end{align*}
Obviously, the proposed idea does not require any detection of the jumps of $x(t)$. 
The following theorem justifies the idea.

\begin{Thm}\label{Theorem: State Estimation}
Suppose that Assumptions~\ref{Ass:standing ass}--\ref{Ass:basic estimation} hold and that there exists a gluing function $\psi$ satisfying \eqref{eq: matching condition of estimation}.
If there is an asymptotic observer for \eqref{System: on Glued Domain} in the sense that $\lim_{t \to \infty} |\zeta(t)-\hat \zeta(t)| = 0$, then, for any $\epsilon>0$, there is $T \geq 0$ (that may depend on the initial conditions of the plant and the observer) such that
\begin{equation*}
|x(t)-\hat x(t)|<\epsilon \qquad\text{for all $t\in\tau_\alpha(\epsilon)\cap(T,\infty)$} 
\end{equation*}
where $\tau_\alpha(\epsilon) = \R_{\geq0} \bs \cup_{i=0}^{N} \BB_{\tau_i}(\alpha(\epsilon))$ with $\tau_i$ being the jump time corresponding to $x(t)$ and $\tau_0=0$.
\end{Thm}

\begin{proof}
Without loss of generality, suppose that $\epsilon$ is sufficiently small so that $ \OO_\DD(\epsilon) \cap \OO_\GG(\epsilon) = \varnothing$.
(This is possible because both $\OO_\DD(\epsilon)$ and $\OO_\GG(\epsilon)$ are contained in the compact set $\EE$ for any $\epsilon > 0$, and two sets $\EE \cap \DD$ and $\EE \cap \GG$ are disjoint under Assumption \ref{Ass:standing ass}.)
Also, noting that $\EE\bs\OO_\DD(\epsilon)$ is compact because $\EE$ is compact and $\OO_\DD(\epsilon)$ is open relative to $\EE$, we see that Lemma \ref{Lem: local embedding of GF} yields $L_1$ such that $|x-x'|\leq L_1|\psi(x)-\psi(x')|$ for all $x$ and $x'$ in $\EE\bs\OO_\DD(\epsilon)$, and $L_2$ such that $|x-x'|\leq L_2|\psi(x)-\psi(x')|$ for all $x$ and $x'$ in $\EE\bs\OO_\GG(\epsilon)$.
Take $L:=\max(L_1,L_2)$.

Since $\lim_{t \to \infty} |\zeta(t)-\hat \zeta(t)| = 0$, there is $T \geq 0$ such that 
$$|\zeta(t) - \hat\zeta(t)| < \frac{\epsilon}{2L} \qquad \text{for all $t>T$}$$
where $T$ may depend on $\epsilon$ and the initial conditions of the plant and the observer.
Thus, we have that
\begin{align*} 
|\zeta(t)-\Pi_{\psi(\EE)}(\hat\zeta(t))| &\leq |\zeta(t)-\hat\zeta(t)| + |\hat\zeta(t)-\Pi_{\psi(\EE)}(\hat\zeta(t))| \\
&\leq |\zeta(t)-\hat\zeta(t)| + |\hat\zeta(t) - \zeta(t)| \\
&< \frac{\epsilon}{L}, \qquad \text{for all $t > T$}.
\end{align*}
Finally, it is seen that, for all $t\in\tau_\alpha(\epsilon)$, both $x(t)$ and $\hat x(t)$ belong to either $\EE\bs\OO_\DD(\epsilon)$ or $\EE\bs\OO_\GG(\epsilon)$, because $x(t)\in \EE\bs(\OO_\DD(\epsilon)\cup\OO_\GG(\epsilon))=(\EE\bs\OO_\DD(\epsilon))\cap(\EE\bs\OO_\GG(\epsilon))$ and $\hat x(t) \in \EE= (\EE\bs\OO_\DD(\epsilon))\cup (\EE\bs\OO_\GG(\epsilon))$ at the time $t$.
As a result, at $t \in \tau_\alpha(\epsilon)$ and $t > T$, we have that $|x(t)-\hat x(t)|\leq L_1|\zeta(t)-\Pi_{\psi(\EE)}(\hat\zeta(t))|<L_1 \epsilon / L \leq\epsilon$ if both $x(t)$ and $\hat x(t)$ belong to $\EE\bs\OO_\DD(\epsilon)$, or that $|x(t)-\hat x(t)|\leq L_2|\zeta(t)-\Pi_{\psi(\EE)}(\hat\zeta(t))|<L_2 \epsilon / L \leq\epsilon$ if both $x(t)$ and $\hat x(t)$ belong to $\EE\bs\OO_\GG(\epsilon)$.
\end{proof}

Theorem \ref{Theorem: State Estimation} can be reinterpreted in a graphical sense: 

\begin{Cor}\label{Cor:Graphcal sense}
Under the assumptions of Theorem \ref{Theorem: State Estimation}, for sufficiently small $\epsilon^*>0$, there exists $T^*>0$ such that
\begin{description}
	\item[(a)] for each $t>T^*$, there exists $s>0$ satisfying that $|(t,x(t))-(s,\hat x(s))|<\epsilon^*$, and;
	\item[(b)] for each $t>T^*$, there exists $s>0$ satisfying that $|(s,x(s))-(t,\hat x(t))|<\epsilon^*$. 
\end{description}
\end{Cor}

\begin{proof}
	See the Appendix.
\end{proof}

\begin{Exam}\label{Ex.GluedObsforBB}
Consider a bouncing ball having a mass $m>0$ with the gravitational constant $\rho>0$. 
The height and velocity of the ball are $x_1$ and $x_2$, respectively. 
The output is the height $x_1$.
Suppose that the coefficient of restitution is $1$.
The system is described with $x = (x_1,x_2)$ by
\begin{alignat}{2}
\begin{bmatrix} \dot x_1 \\ \dot x_2 \end{bmatrix} &= f(x) = \begin{bmatrix} x_2 \\ -\rho \end{bmatrix} & \quad & \text{when $x\in \CC \subset \R^2$}, \notag \\
{x}^+ &= g(x)= -x & & \text{when $x\in \DD \subset \R^2$}, \label{exeq:BBwithOutput} \\
y &= h(x) = x_1, \notag
\end{alignat}
where $\CC:=\{( x_1, x_2)\in\R^2: (x_1\geq 0) \;\&\; (|x|\neq0)\}$ and $\DD:=\{(0, x_2)\in\CC : x_2<0\}$. 
Let $\EE:=\{x\in \CC: \underline\delta \leq m\rho x_1 + (1/2) m x_2^2 \leq \overline \delta\}$ where $\overline \delta>\underline\delta>0$. 
The system satisfies Assumptions~\ref{Ass:standing ass}--\ref{Ass:basic estimation} with $x(0) \in \EE$, and the geometry of the domain helps to find gluing functions.
One idea is to glue both $\DD$ and $\GG$ on the negative $x_1$ axis by doubling the angle of the vector from the positive $x_1$ axis. As a result, the glued domain $\CC^\psi$ becomes $\R^2\backslash\{0_2\}$.
The conference version \cite{JS14} of this paper illustrates this idea.
However, the glued system is not very simple.
Another idea for finding a gluing function is to embed the domain $\CC \subset \R^2$ into $\R^3$ by gluing $\DD$ and $\GG$ together by dragging both $\DD$ and $\GG$ onto the vertical axis (see Figure \ref{fig:figure}). 
One of the functions realizing this idea is $\psi:{\CC}\rightarrow \R^3$ defined as $x=(x_1,x_2) \mapsto \zeta = (\zeta_1,\zeta_2,\zeta_3) = (x_1^2, 2x_1x_2, 2x_2^2+4\rho x_1)$.
Note that
$\psi$ is a gluing function satisfying the vector field matching condition \eqref{eq: matching condition of estimation}. 
The inverse gluing function is
$$\psi^{-1}(\zeta)=\begin{bmatrix}
\sqrt{\zeta_1},& {\overline{\rm sgn}}(\zeta_2)\sqrt{\frac{\zeta_3-4\rho\sqrt{\zeta_1}}{2}}
\end{bmatrix}^\top \quad \hbox{for all } \zeta \in \CC ^\psi$$
where $\overline{{\rm sgn}}(\zeta_2)$ is $1$ if $\zeta_2 \geq 0$ and $-1$ if $\zeta_2 < 0$.
Through the gluing function, we obtain the functions $f^\psi(\zeta)=[\zeta_2, \zeta_3-6\rho\sqrt{\zeta_1}, 0]^\top$ and $h^\psi(\zeta)=\sqrt{\zeta_1}$, so that the glued system is written as a continuous-time dynamical system:
\begin{align}
\begin{split}\label{Example:GluedBB}
\dot \zeta&=\begin{bmatrix}
0 & 1 & 0\\0 & 0 & 1\\0 & 0 & 0\\
\end{bmatrix}\zeta +\begin{bmatrix}
0  \\
-6\rho \\
0\\
\end{bmatrix} \sqrt{\zeta_1}=:A\zeta+B\sqrt{\zeta_1},\\
y&=\sqrt{\zeta_1},
\end{split}
\end{align}  
for $\zeta \in \CC^\psi$. 
We propose an observer for \eqref{Example:GluedBB} as
\begin{align}\label{eq:ex Luenberger}
\dot {\hat \zeta}=A\hat\zeta+L(C\hat \zeta-y^2)+By,
\end{align}
where $C=\begin{bmatrix}
1 & 0 & 0
\end{bmatrix}$ and $A+LC$ is Hurwitz. 
Let $\zeta_e:=\hat \zeta - \zeta$. 
Then, the dynamics of $\zeta_e$ results in 
$$\dot \zeta_e=(A+LC)\zeta_e,$$
because $y^2=C\zeta$. 
Therefore, $\zeta_e$ converges to zero exponentially and, by Theorem \ref{Theorem: State Estimation}, we obtain an estimate 
\begin{align}\label{exeq:etimate}
\begin{split}
\hat x &= \psi^{-1}(\Pi_{\psi(\EE)}(\hat{\zeta})) = \psi^{-1}(\bar{\zeta_1},\bar{\zeta_2},\bar{\zeta_3}) \\
&= \left (\sqrt{\bar\zeta_1},~ {\overline {\rm sgn}}(\bar \zeta_2)\sqrt{\frac{\bar\zeta_3-4\rho\sqrt{\bar \zeta_1}}{2}}\right )
\end{split}
\end{align}
where $\Pi_{\psi(\EE)}(\hat{\zeta})$ finds a vector $\bar{\zeta} \in \psi(\EE)$ that has the minimum distance from $\hat{\zeta}$.
Figure \ref{Figure: Estimation error:BB} shows a simulation result.
\hfill $\End$
\end{Exam}

\begin{figure}
	\centering
	\includegraphics[width=.8\columnwidth]{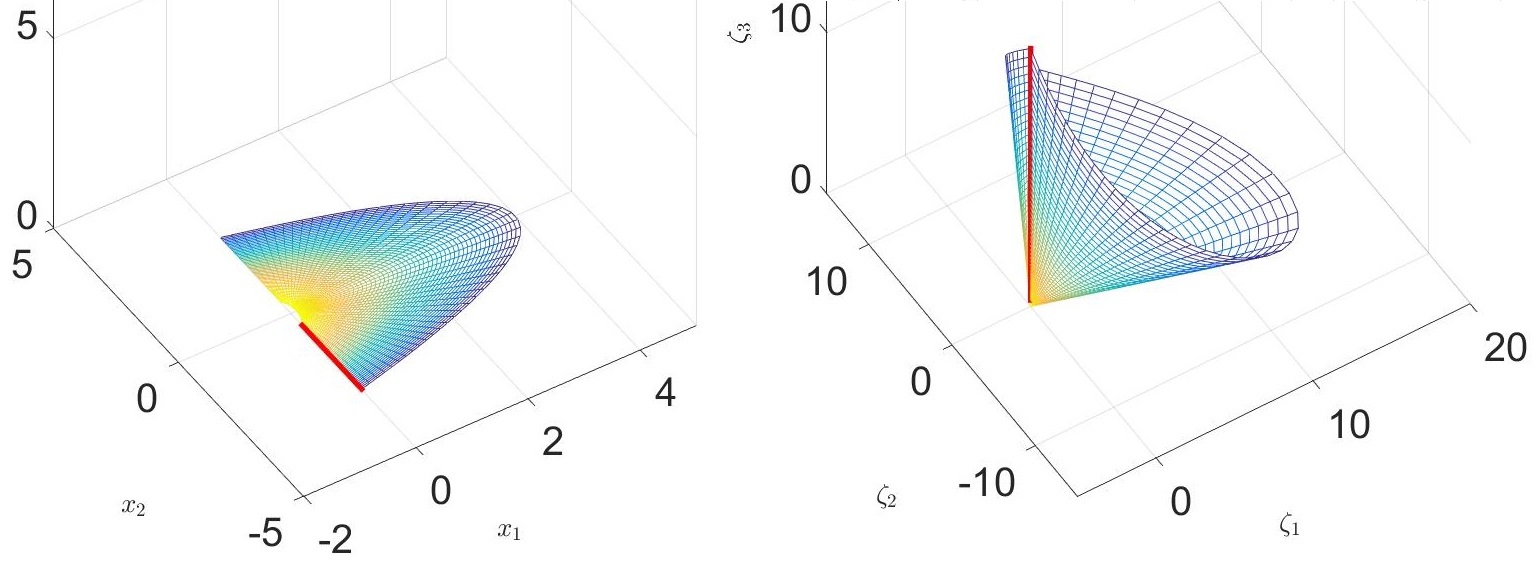}	
	\caption{Visualization of the gluing in Example \ref{Ex.GluedObsforBB}; the set $\DD$ is in red.} 
	\label{fig:figure}
\end{figure}

\begin{figure}[t]
	\begin{minipage}{.5\linewidth}
		\centering
		\includegraphics[width=1\columnwidth]{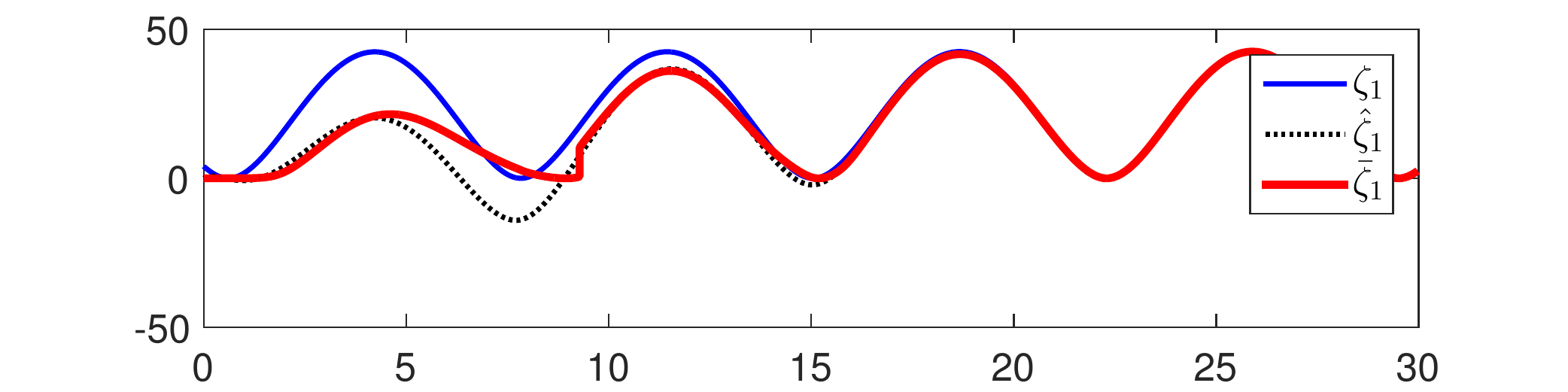}
		\includegraphics[width=1\columnwidth]{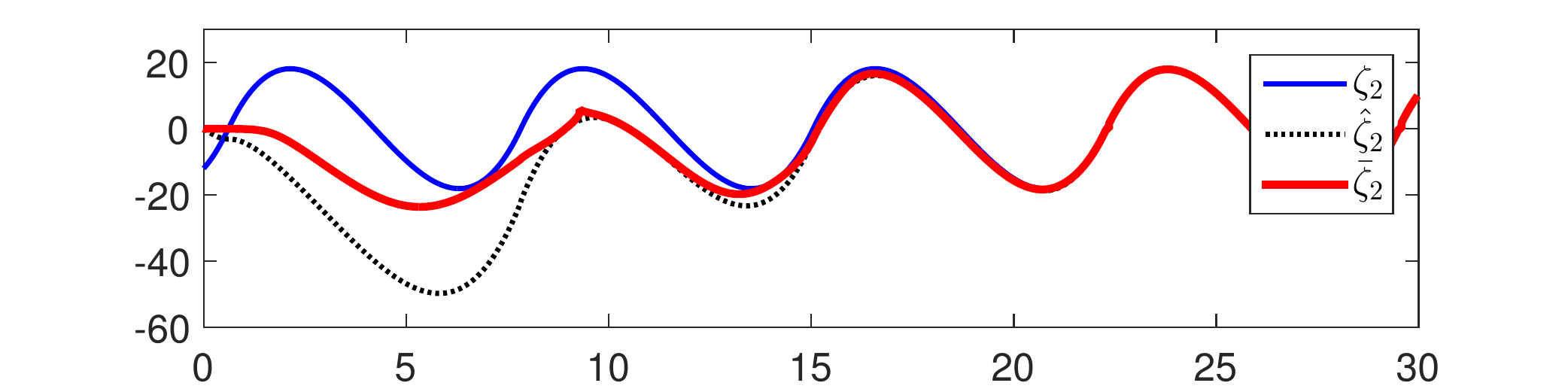}
		\includegraphics[width=1\columnwidth]{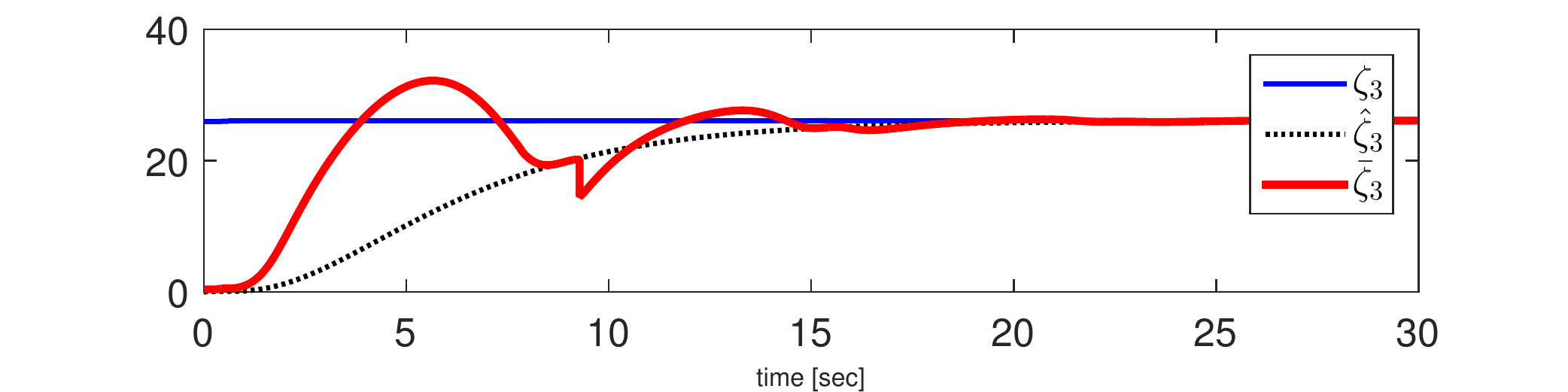}
	\end{minipage}~
	\begin{minipage}{.5\linewidth}
		\centering
		\includegraphics[width=1\columnwidth]{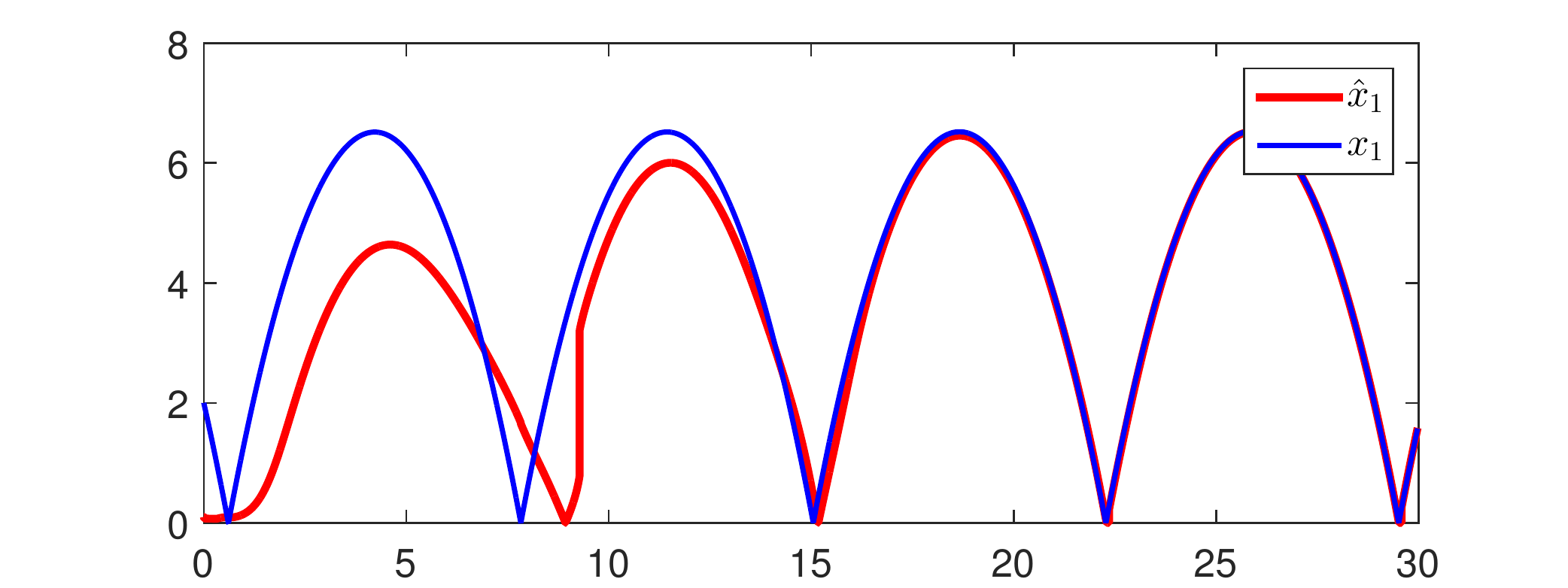}
		\includegraphics[width=1\columnwidth]{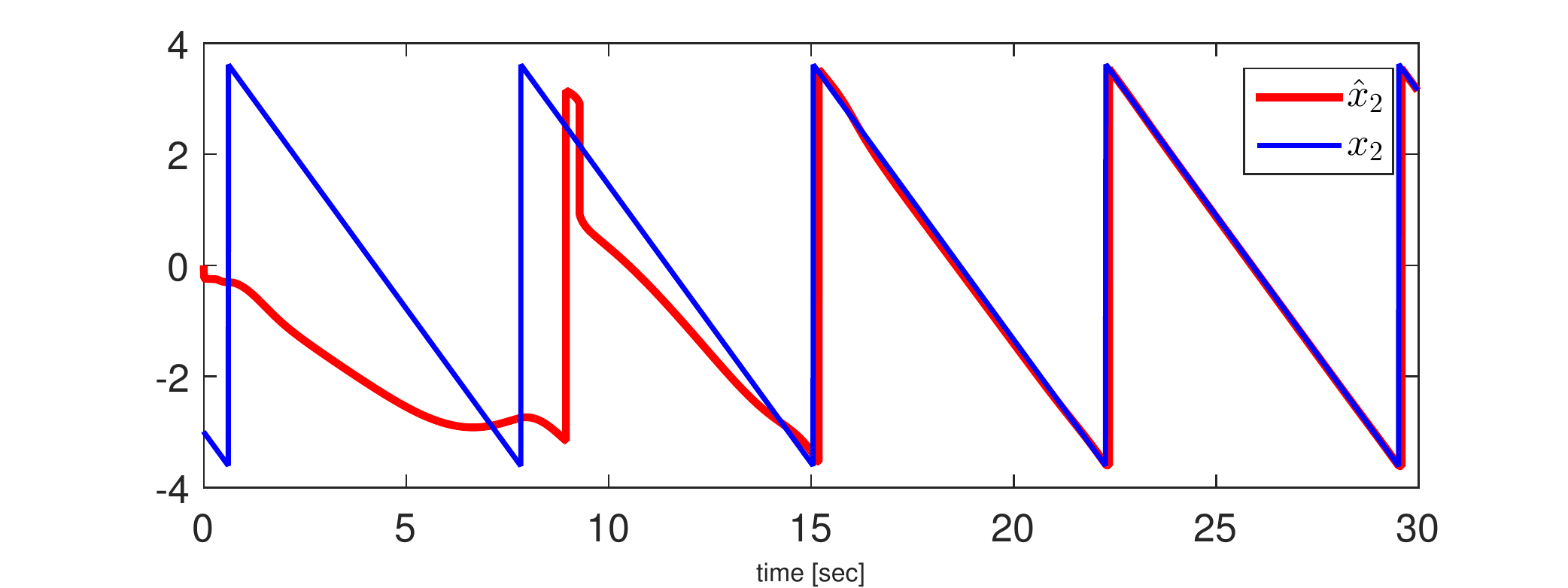}
	\end{minipage}
	\caption{A simulation result of Example \ref{Ex.GluedObsforBB} when $\rho=1$ and $x_0=(2,-3)$. The left three figures illustrate the state $(\zeta_1,\zeta_2,\zeta_3)$ of \eqref{Example:GluedBB} in blue, the estimate $(\hat\zeta_1,\hat\zeta_2,\hat\zeta_3)$ obtained from \eqref{eq:ex Luenberger} in black, and the projected estimate $(\bar\zeta_1,\bar\zeta_2,\bar\zeta_3)$ onto $\psi(\EE)$ in red. The right two figures depict the state $(x_1,x_2)$ of \eqref{exeq:BBwithOutput} in blue and the estimate $(\hat x_1,\hat x_2)$ obtained from \eqref{exeq:etimate} in red.} 
	\label{Figure: Estimation error:BB}
\end{figure}

\subsection{Observer Designs}

Since the glued system is a continuous-time system, one may employ any conventional observer design technique for continuous-time systems. However, many of them require additional properties for $f^\psi$ and $h^\psi$, e.g., linearity, Lipschitz continuity, triangular structure, and so on.
Since some of those conditions may depend on the selection of the gluing function $\psi$, we discuss the conditions that lead to such properties and propose two observer designs.

\smallskip

\subsubsection{Observer with Linearized Error Dynamics}\label{subsec:particular}
We first propose a condition for the existence of a gluing function that makes the glued system become a linear system with output injection; a form used for observer designs in, e.g., \cite{Kel87,JB04}. 
To simplify the presentation, we only consider the case when the output is scalar, and both $f$ and $h$ are smooth.
Then, the following assumption and theorem yield a gluing function and an observer design.

\begin{Assn}\label{Ass:immersion}
There are an integer $m \ge k$ and a smooth injective function $\phi: h(\CC) \to \R$ such that, with $h^*(x) := \phi(h(x))$, \\ 
\noindent{\bf (I1)} $L^i_f h^*(x)=L^i_f h^*(x_g)|_{x_g=g(x)}$ for $1\leq i\leq m$ and for all $x\in \DD$;\\
\noindent{\bf (I2)} there exist smooth solutions $a_1(h^*),\dots,a_m(h^*)$ to the differential equation
$$L_f^m h^*=a_m(h^*)+L_f a_{m-1}(h^*)+\cdots+L^{m-1}_f a_1(h^*).$$
\hfill $\End$
\end{Assn}

The role of $\phi$ is to re-define the output function $h$, by which the condition is slightly extended from the case when $\phi$ is an identity so that $h^*(x) = h(x)$. 

\begin{Thm}\label{Thorem: immersion observer}
Suppose that Assumptions~\ref{Ass:standing ass}--\ref{Ass:immersion} hold. 
If 
$$\psi(x) := \begin{bmatrix} h^*(x) \\ L_fh^*(x)-a_1(y^*)|_{y^*=h^*(x)} \\ \vdots \\ L^{m-1}_fh^*(x)-\sum^{m-1}_{i=1} L_f^{m-1-i} a_{i}(y^*)|_{y^*=h^*(x)} \end{bmatrix}$$
satisfies (G2), (G4), and (G5), then $\psi$ is a gluing function satisfying \eqref{eq: matching condition of estimation}. 
Furthermore, the glued system is given by
\begin{align}\label{eq:gsys_errlin}
\begin{split}
\dot \zeta_1 &= \zeta_2+a_1(\zeta_1)\\
&\vdots \\
\dot \zeta_{m-1} &= \zeta_m+a_{m-1}(\zeta_1)\\
\dot \zeta_m &= a_m(\zeta_1)\\	
y^*&=\phi(y)=\zeta_1 .
\end{split}
\end{align}
\end{Thm}

\begin{proof}
Since $a_1,\dots,a_{m-1}$, $f$, and $h^*$ are smooth, (G3) holds. 
In addition, by (E3) and (I1), it can be seen that (G1) follows and $\psi$ satisfies \eqref{eq: matching condition of estimation}. 
By the construction, the glued system \eqref{eq:gsys_errlin} is obtained.
In particular, the equation $\dot \zeta_m=a_m(\zeta_1)$ follows from (I2).
\end{proof}

The glued system in Theorem~\ref{Thorem: immersion observer} is a linear system up to output injection. 
Hence, we can design an observer for $\HH\pp^{h}$ as
$$\dot {\hat \zeta} = A \hat \zeta+L(C\hat \zeta-y^*) +a(y^*)$$
where
$$A := \begin{bmatrix}
	0\! & 1 & \cdots\! &  \!0\\
	\vdots &  \cdots\!& \ddots &  \!0\\
	0\! & \cdots\!& 0 &  \!1\\
	0\!& \cdots\!& 0 &  \!0\\
\end{bmatrix}, \; 
C^\top := \begin{bmatrix}	1\\ 0 \\ \vdots \\ 0
\end{bmatrix}, \;
a(\cdot) := \begin{bmatrix}
	a_1(\cdot)\\
	a_2(\cdot)\\
	\vdots\\
	a_m(\cdot)\\
\end{bmatrix}$$
and $A+LC$ is Hurwitz. 
Then, the error convergence is easily shown since the error dynamics is linear.

\smallskip

\subsubsection{Observer for Lipschitz Continuous Systems \cite{GK03}}
If the glued system \eqref{System: on Glued Domain} is observable in the sense of \cite{GK03} and $f^\psi$ and $h^\psi$ are Lipschitz continuous, then one can employ the observer design proposed in \cite{GK03}.
While the observability can be checked even in the original domain (for details, see \cite[Section IV]{JS14}, omitted due to the page limit), verification of Lipschitz continuity of $f^\psi$ and $h^\psi$ is often not easy because it is tedious to obtain analytic forms of $f^\psi$ and $h^\psi$ with respect to $\zeta$ by \eqref{eq: flow after gluing} and \eqref{eq:output after gluing}.
Instead, we propose a condition that guarantees Lipschitz continuity of the glued system, which can be checked without explicitly obtaining the glued system.
Consequently, the numerical construction of the observer, presented in \cite{GK03}, can be performed not by $f^\psi$ and $h^\psi$ but by $f$ and $h$.

\begin{Thm} \label{Theorem: Lip conti}
Under Assumptions \ref{Ass:standing ass}--\ref{Ass:basic estimation}, suppose that there is a $C^2$ gluing function $\psi:\CC\to\R^k$ satisfying \eqref{eq: matching condition of estimation}, and that $h$ is $C^1$.
Then, the functions $f^\psi$ and $h^\psi$ defined in \eqref{eq: flow after gluing} and \eqref{eq:output after gluing} are Lipschitz continuous on $\psi(\EE)$.
\end{Thm}

\begin{proof}
	See the Appendix.
\end{proof}

We illustrate the previous discussion through the hybrid ripple model studied by \cite{Za14} in the following example. The observer designed in \cite{Za14} requires the information of jumping times or an algorithm to estimate the jumping times.
On the other hand, the proposed observer does not need such requirements.

\begin{Exam}\label{ExSecond}
Consider a hybrid system given by
\begin{align}\label{Example:linear hybrid}
\begin{split}
\dot x &= \begin{bmatrix} 0 & 1 \\ -1 & 0 \end{bmatrix}x =: Ax \quad \text {when $x\in\CC$} \\
x^{+} &= \begin{bmatrix} 1 & 0 \\ 0 & -1 \end{bmatrix}x =: Jx \quad \text {when $x\in\DD$} \\
y &= \begin{bmatrix} 1 & 0 \end{bmatrix}x =: Hx
\end{split}
\end{align}
where $\CC=\{x\in \R^2 : (|x|>0) \;\&\; ( h_1 x\geq 0 ) \;\&\; ( h_2 x\geq 0) \}$, $\DD=\{x\in \CC:h_1 x= 0 \}$, $h_1=[\sqrt 3,1]$, and $h_2=[\sqrt 3,-1]$. 
The output matching condition (E3) holds because $HJx=x_1=Hx$. 
When we take $\EE=\{x\in \CC : 1\leq|x|\leq 3 \}$, the other conditions in Assumptions~\ref{Ass:standing ass}--\ref{Ass:basic estimation} are easily verified. 

From the intuition, we design the gluing function $\psi(x)$ of \eqref{Example:linear hybrid} such that it triples the angle of $x$ in the polar coordinates.
Indeed, we define $\psi:\CC\rightarrow\R^2$ as $(x_1,x_2) = (\rho \cos \theta, \rho \sin \theta) \mapsto (\zeta_1,\zeta_2) = (\rho \cos 3\theta, \rho \sin 3\theta)$, which has the realization:
\begin{align*}
\psi( x) &= \bigg[\frac{4x^3_1}{| x|^2}-3x_1,~-\frac{ 4x_2^3}{| x|^2}+3x_2\bigg]^\top.
\end{align*}
It is easy to check that $\psi$ is a gluing function satisfying \eqref{eq: matching condition of estimation} and the assumptions of Theorem \ref{Theorem: Lip conti}.
Then, by the theorem, the glued system of \eqref{Example:linear hybrid} by $\psi$ is Lipschitz continuous on $\EE$. 
In fact, an analytic form of the glued system can also be obtained in this case as, for all $\zeta \in \CC^\psi = \R^2 \bs \{0_2\}$, 
\begin{align*} 	
\dot \zeta &=f^\psi(\zeta)=\begin{bmatrix}
0&3 \\ -3&0
\end{bmatrix}\zeta ,\\
y&=h^\psi(\zeta)=\hbox{Re}\left(\sqrt[3]{\frac{\zeta_1+i|\zeta_2|}{|\zeta|}}\right)|\zeta|
\end{align*}
where $i$ is the imaginary unit, and it is verified that this system is Lipschitz continuous on $\EE$. 
Moreover we construct an observer for the glued system by using the output trajectory of \eqref{Example:linear hybrid} according to the numerical design procedure provided in \cite{GK03}. 
The details can be found in \cite{GK03,AL99}. 
A simulation result is illustrated in Figure \ref{Figure: Estimation error:linear}.
\hfill $\End$
\end{Exam}

\begin{figure}[t]
\centering
\includegraphics[width=1\columnwidth]{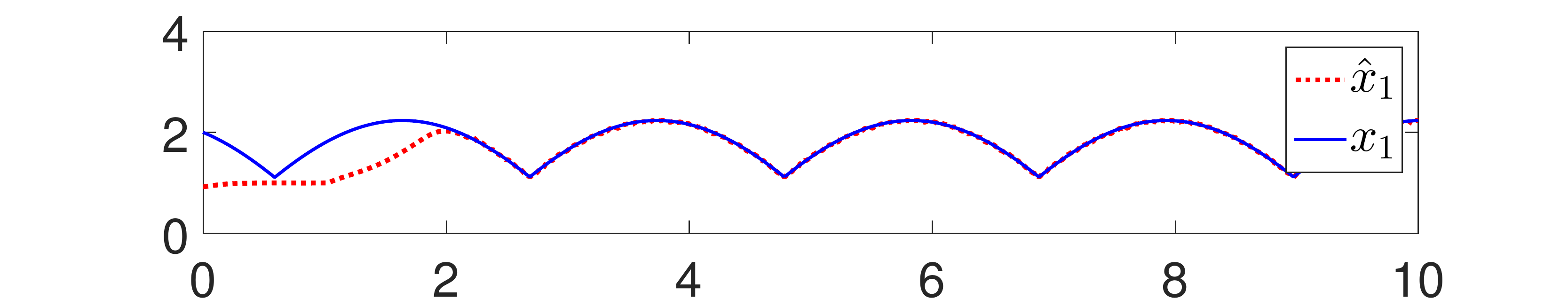}
\includegraphics[width=1\columnwidth]{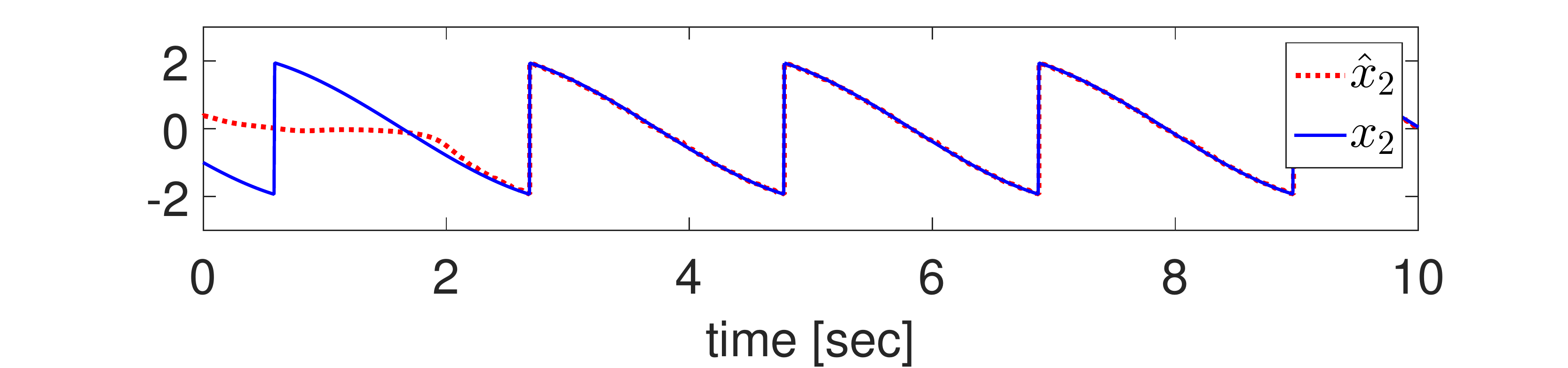}
\caption{Simulation result of Example \ref{ExSecond}: the state $(x_1,x_2)$ is in blue and its estimate $(\hat x_1,\hat x_2)$ is in red.} 
\label{Figure: Estimation error:linear}
\end{figure}

\section{Tracking Control by Gluing}\label{Section IV}

In this section, we consider asymptotic state tracking problem. 
The goal is to obtain a controller that steers the state trajectory to track a given state reference. 

\subsection{Standing Assumptions}

To consider the tracking control problem, we impose the following assumptions on the plant $\HH\pp=(\CC,f,\DD,g)$ and the reference $r(t)$.

\begin{Assn}\label{Ass: basic tracking}
\noindent{\bf (C1)} $r(t)$ is a state trajectory of an execution of $\HH\pp$ which is infinite in $t$-direction starting at $r(0) \in \CC$ under a piecewise continuous input $u=u\rr(t) \in \R^p$;\\
\noindent{\bf (C2)} there is a compact subset $\RR$ of $\CC$ such that $r(t)\in \RR$ for all $t\geq0$ and, for some $\mu>0$, it holds that 
$$\left\{ \begin{array}{ll} \nabla r_\DD (x)\cdot f(x,u)>\mu& \hbox{for all } (x,u)\in(\DD\cap \RR)\times \R^p, \\
\nabla r_\GG (x)\cdot f(x,u)>\mu& \hbox{for all } (x,u)\in(\GG\cap \RR)\times \R^p. \end{array}\right.$$
\hfill $\End$
\end{Assn}

The condition (C1) is necessary because, when $x_0 = r(0)$, one expects $x(t) = r(t)$ for all future time with some input.
On the other hand, (C2) asserts that, for any inputs, the state $x(t)$ only jumps on $\DD$ and then flows into the interior of $\CC$ on $\GG$.
Now, the objective is to construct a feedback control law $u=u\cc(u\rr,r,x)$ that makes $x(t)$ track $r(t)$.

\subsection{Tracking Control}

Suppose that there exists a gluing function $\psi$ for $\HH\pp$.
As shown in the previous section, if the matching condition 
\begin{equation}\label{eq: matching condition for tracking}
\dd\psi(x) f(x,u) = \dd\psi(g(x)) f(g(x),u), \; \forall (x,u)\in \DD\times \R^p
\end{equation}
holds, we can take $\dd\psi(x) f(x,u)$ as the vector field at $\psi(x) \in \CC^\psi$.
Therefore, we obtain the glued system as
\begin{align}\label{System: Glued system with input}
\dot \zeta =\dd\psi (\psi^{-1}(\zeta)) f(\psi^{-1}(\zeta),u)=: f^\psi(\zeta,u) .
\end{align}
Note that $f^\psi$ is continuous by (A2), (G3), \eqref{eq: matching condition for tracking}, and Lemma~\ref{Lem: local embedding of GF}. 

Next, we define a {\em glued reference} by $\psi$ as 
\begin{align}\label{eq: glued reference trajectory}
\zeta\rr(t):=\psi(r(t))\quad \hbox{for all $t\geq 0$}.
\end{align}
Then, $\zeta\rr(t)$ is continuous and it is a solution to \eqref{System: Glued system with input} when $u=u\rr(t)$. 
For the continuous-time system \eqref{System: Glued system with input} and the continuous reference \eqref{eq: glued reference trajectory}, one may design a feedback control of the form
\begin{equation}\label{System: Controller for Glued system}
u = u\cc^\psi(u\rr,\zeta\rr,\zeta).
\end{equation}
In this case, the closed-loop system becomes
\begin{equation*}
\dot \zeta = f^\psi(\zeta,u_{\rm c}^\psi(u\rr(t),\zeta\rr(t),\zeta)).
\end{equation*}

\begin{Assn}\label{Ass:controller}
There are a gluing function $\psi$ of $\HH\pp$ satisfying \eqref{eq: matching condition for tracking} and a local tracking controller \eqref{System: Controller for Glued system} such that \\
\noindent{\bf(C3)} $u^\psi_{\C}: \R^p \times \CC^\psi \times \CC^\psi \to \R^p$ is locally Lipschitz;\\
\noindent{\bf(C4)} for each $\gamma>0$, there exists an open neighborhood $\VV^\psi$ of $\zeta_{\rm r}(0)$ in $\CC^\psi$ such that, if $\zeta(0)\in \VV^\psi$, then
$$|\zeta(t)-\zeta_{\rm r}(t)|<\gamma \hbox{ for all } t\geq 0 \hbox{ and } \lim_{t\to \infty} |\zeta(t)-\zeta_{\rm r}(t)|=0.$$   \hfill $\End$ 
\end{Assn}

Through the controller in Assumption~\ref{Ass:controller}, we propose the tracking controller for $\HH\pp$ as
\begin{equation*}
u = u\cc(u\rr,r,x) := u^\psi\cc(u\rr,\psi(r),\psi(x)).
\end{equation*}
Then, the solution $x(t)$ of the closed-loop system starting near $r(0)$ is well-defined for all $t\geq0$ by (C2) and (C4). In addition, it is unique since, by (C3), the flow map $f(x,u\cc^\psi(u_r(t),\psi(r(t)),\psi(x)))$ is locally Lipschitz with respect to $x$ and piecewise continuous in $t$ and, by (C2), the solution must jump on $\DD$.

\begin{Thm}\label{Thm:contoller Thm}
Under Assumptions~\ref{Ass:standing ass},~\ref{Ass: basic tracking}, and~\ref{Ass:controller}, there exist an open set $\VV$ in $\CC$ containing $r(0)$, and a class-$\KK$ function $\alpha$ satisfying that: for $x_0 \in \VV $ and for any given $\epsilon >0$, there exists $T\geq0$ such that
 \begin{align}\label{eq: tracking result of thm}
 |x(t)-r(t)|<\epsilon \qquad\text{for all $t\in\tau_\alpha(\epsilon)\cap(T,\infty)$} 
 \end{align}
where $\tau_\alpha(\epsilon) = \R_{\geq0} \bs \cup_{i=0}^{N} \BB_{\tau_i}(\alpha(\epsilon))$ with $\tau_i$ being the jump time corresponding to $r(t)$ and $\tau_0=0$.
\end{Thm}

\begin{proof}
Since $\psi$ is continuous and $\RR$ is compact, the image $\psi(\RR)\subset \CC^\psi$ is also compact. 
Then, the set $\WW^\psi:=\{\zeta\in \CC^\psi : d_{\psi(\RR)}(\zeta) \leq \gamma^* \}$ becomes compact with sufficiently small $\gamma^* > 0$.
By (C4), there is an open set $\VV^\psi$ such that $\zeta(t) \in \WW^\psi$ for all $t \ge 0$ when $\zeta(0) \in \VV^\psi$.
Hence, by (G5), we have the compact set $\WW:=\psi^{-1}(\WW^\psi)$ containing $x(t)$ and $r(t)$ for all $t\geq0$.

Applying Lemma~\ref{Lem: Existence of alpha and beta}, we obtain a class-$\KK$ function $\alpha$ satisfying that, for a sufficiently small $\epsilon>0$, $r(t) \notin \{x\in\WW : d_{\GG\cap \WW}(x) < \epsilon\} \cup \{x\in\WW : d_{\DD\cap \WW}(x) < \epsilon\}$, $\forall t \in \tau_\alpha(\epsilon)$.
Let $\VV:=\psi^{-1}(\VV^\psi)$. 
Then, it is open in $\CC$ because $\psi$ is continuous and, for $x_0 \in \VV$, $x(t)$ belongs to the compact set $\WW$. 
Note that $r(t)$ also stays in the compact set $\WW$. 
Therefore, one obtains \eqref{eq: tracking result of thm} similarly as in the proof of Theorem~\ref{Theorem: State Estimation}.
\end{proof}

\subsection{Controller Design}\label{Section: Nomal form}

Finding gluing functions of $\HH\pp$ satisfying \eqref{eq: matching condition for tracking} may not be a trivial task. 
If \eqref{eq: matching condition for tracking} does not hold, the system in the glued domain may have discontinuous vector fields (even though state jumps disappear), and may be regarded as a (state-triggered) switched system. 
Nevertheless, there is a possibility to counteract the residual discontinuity of the vector field by feedback control in some cases, which is reminiscent of the fact that a feedback cancels residual nonlinearity in the feedback linearization technique for nonlinear systems.
This possibility is exploited in this subsection.

Consider a plant whose flow map is modeled as an input affine form:
\begin{align*}
f(x,u) = a(x)+b(x)u
\end{align*} 
where $a$ and $b$ are smooth. 
Then, equation \eqref{eq: matching condition for tracking} holds when the conditions
\begin{alignat}{2}
\dd\psi(x)a(x)&=\dd\psi(g(x))a(g(x))\quad &&\hbox{for all $x\in\DD$,} \label{eq: matching condition for tracking f}\\
\dd\psi(x)b(x)&=\dd\psi(g(x))b(g(x))\quad &&\hbox{for all $x\in\DD$,} \label{eq: matching condition for tracking b}
\end{alignat}
are satisfied.
However, since there are some systems that do not satisfy these conditions, let us consider a relaxation of them through a feedback.
Suppose that there exist a $C^1$ matrix function $\gamma:\CC \to \R^{p\times p}$ such that the matrix $\gamma(x)$ is invertible for all $x \in \CC$, and a $C^1$ function $\kappa:\CC\to \R^p$ such that, for all $x\in\DD$,
\begin{align}
&\dd \psi(x)b(x)\gamma(x) = \dd \psi(g(x))b(g(x))\gamma(g(x)), \label{eq: relaxed matching condition for tracking b} \\
&\dd \psi(x)a(x) + \dd \psi(x)b(x)\gamma(x)\kappa(x) \notag \\
&\qquad = \dd \psi(g(x))a(g(x)) + \dd \psi(x)b(x)\gamma(x) \kappa(g(x)) . \label{eq: relaxed matching condition for tracking a}
\end{align}
It is seen in \eqref{eq: relaxed matching condition for tracking b} and \eqref{eq: relaxed matching condition for tracking a} that the mismatches yielding violation of \eqref{eq: matching condition for tracking f} and \eqref{eq: matching condition for tracking b} are absorbed in $\gamma$ and $\kappa$.
The glued system is written as
\begin{align*}
\dot \zeta &= \dd\psi(\psi^{-1}(\zeta)) a(\psi^{-1}(\zeta)) + \dd\psi(\psi^{-1}(\zeta)) b(\psi^{-1}(\zeta)) u \\
&=: a^{\psi}(\zeta) + b^{\psi}(\zeta) u
\end{align*}
in which, $a^{\psi}$ and $b^{\psi}$ are not continuous unless \eqref{eq: matching condition for tracking f} and \eqref{eq: matching condition for tracking b} hold.
Now, consider a feedback control
\begin{equation}\label{eq:v}
u = \gamma(\psi^{-1}(\zeta)) ( v + \kappa(\psi^{-1}(\zeta)) )
\end{equation}
where $v$ is a new input.
Then, it can be shown that, in the closed-loop system
\begin{align}
\dot \zeta &= a^{\psi}(\zeta) + b^{\psi}(\zeta) \gamma(\psi^{-1}(\zeta)) \kappa(\psi^{-1}(\zeta)) + b^{\psi}(\zeta) \gamma(\psi^{-1}(\zeta)) v \notag \\
&=: a^{\psi}_{\kappa\gamma}(\zeta) + b^{\psi}_{\gamma}(\zeta) v, \qquad \zeta \in \CC^\psi, \label{eq: feedback continuous form}
\end{align}
two functions $a^{\psi}_{\kappa\gamma}$ and $b^{\psi}_{\gamma}$ are continuous with respect to $\zeta$, thanks to \eqref{eq: relaxed matching condition for tracking b}, \eqref{eq: relaxed matching condition for tracking a}, Lemma \ref{Lem: local embedding of GF}, and (G3).
Note that \eqref{eq:v} itself may not be continuous since $\gamma(\psi^{-1}(\zeta))$ and $\kappa(\psi^{-1}(\zeta))$ may have discontinuity at $\zeta \in \psi(\DD)$. 
It is seen that $\zeta\rr(t)$ is a solution to \eqref{eq: feedback continuous form} when $v = (\gamma(\psi^{-1}(\zeta\rr(t))))^{-1}(u\rr(t)-\kappa(\psi^{-1}(\zeta\rr(t)))) =: v\rr(t)$. 

Since $a_{\kappa\gamma}^\psi$ and $b^\psi_\gamma$ are continuous, we proceed as in the previous subsection to find a control law $v^\psi_{\C}$ for \eqref{eq: feedback continuous form} yielding the closed-loop system
$$\dot \zeta = a_{\kappa\gamma}^\psi(\zeta) + b^\psi_\gamma(\zeta) v^\psi_{\C}(v\rr,\zeta\rr,\zeta).$$
Then, the tracking controller for $\HH\pp$ is obtained as
\begin{align}\label{System: no jump controller}
\begin{split}
&u_{\C}(u\rr,r,x) = \gamma(x) \left( v^\psi_{\C}(v\rr,\psi(r),\psi(x)) + \kappa(x) \right) \\
&=\gamma(x) \left( v^\psi_{\C}\left( \gamma(r)^{-1} (u\rr\! - \kappa(r)), \psi(r), \psi(x)\right) + \kappa(x) \right) .
\end{split}
\end{align}

Because of discontinuities in the feedback of \eqref{eq:v}, the feedback control $v\cc^\psi$, instead of \eqref{eq:v}, should play the role of $u\cc^\psi$ in (C3) and Theorem \ref{Thm:contoller Thm}.

\begin{Exam}\label{Example: Number}
Consider a system 
\begin{align}\label{Example: linear}
\begin{bmatrix} \dot x_1 \\ \dot x_2 \end{bmatrix} &= \begin{bmatrix}
a_{11} &a_{12}\\a_{21} &a_{22}
\end{bmatrix} \begin{bmatrix} x_1 \\ x_2 \end{bmatrix} + \begin{bmatrix} 0 \\ b \end{bmatrix} u =: Ax + Bu~~~~~~
\end{align}
\vspace{-.6cm}	
\begin{alignat}{2}
& &&\hbox{ when } x\in \CC=\{x\in \R^2: (x_1\geq 0) \;\&\; (|x|>0)\},\notag \\
x^+&=-x &&\hbox{ when } x\in \DD=\{x\in \CC: (x_1=0) \;\&\; (x_2 < 0)\},\notag
\end{alignat}
where $a_{12}>0$ and $b\neq0$. 
Suppose that the reference $r(t)$ satisfies Assumption~\ref{Ass: basic tracking}.

In order to facilitate the search for a gluing function, we consider an alternative system in $\R^3$ whose behavior looks the same as \eqref{Example: linear}.
The system is given by, with $\bar x = (x, p)$ where $p$ is a mode variable taking value of either $+1$ or $-1$, 
\begin{alignat}{2}
\dot {\bar x} &=\bar f(\bar x)+\bar b(\bar x)u&\ \label{Example: linear-2} \\
&:= \begin{bmatrix}
	A & 0\\
	0 & 0
\end{bmatrix}\bar x+\begin{bmatrix}
	B\\0
\end{bmatrix}
u \qquad&&\hbox{when } \bar x\in \bar\CC:=\CC\times \{-1,1\}, \notag\\
	{\bar x}^+&=\bar g (\bar x):=-\bar x &&\hbox{when } \bar x\in \bar\DD:=\DD\times \{-1,1\}.\notag 
\end{alignat}
Note that the first two elements of the state trajectory of \eqref{Example: linear-2} and the state trajectory of \eqref{Example: linear} coincide. 
		
We take the reference $\bar r(t):=(r(t),p\rr(t))$, which is a state trajectory of \eqref{Example: linear-2} with $\bar r(0) := (r(0),1)$ and $u = u\rr(t)$. 
For this system, we take a gluing function $\psi: \bar{\CC} \to \R^2$ defined as $\psi(\bar x)=px=:\zeta$.
With this gluing function, the condition \eqref{eq: matching condition for tracking f} holds, but the condition \eqref{eq: matching condition for tracking b} does not (indeed, for $\bar x \in \bar\DD$, $\dd\psi(\bar x)\bar b(\bar x)=pB$ and $\dd\psi(\bar g(\bar x))\bar b(\bar g(\bar x))=-pB$ are not the same).
	
Take $\gamma(\bar x) = p$ (and $\kappa(\bar x) = 0$). 
Note that $\gamma(\bar x)$ is continuously differentiable and non-zero for all $\bar x \in \bar \CC$. 
Then, because $\dd\psi(\bar x)\bar b(\bar x)\gamma(\bar x)
=\dd\psi(\bar g(\bar x))\bar b(\bar g(\bar x))\gamma(\bar g(\bar x))=B$ for $\bar x \in \bar\DD$, \eqref{eq: relaxed matching condition for tracking b} is satisfied (and \eqref{eq: relaxed matching condition for tracking a} as well).

We obtain that $\bar a^\psi(\zeta) = pAx = A\zeta$. 
In addition, since the inverse gluing function in \eqref{eq: Inverse Function of Gluing Function} is obtained as
$$\psi^{-1}(\zeta):= \overline{\rm Sgn} (\zeta)\begin{bmatrix}
	 \zeta_1\\ \zeta_2\\ 1
\end{bmatrix} \; \hbox{for $\zeta=(\zeta_1, \zeta_2)\in\bar\CC^\psi= \R^2 \bs \{0_2\}$}$$
where $\overline{\rm Sgn}(\zeta)$ is ${\zeta_1}/{|\zeta_1|}$ if $\zeta_1 \neq 0$ and ${\zeta_2}/{|\zeta_2|}$ otherwise, we have that $\gamma(\psi^{-1}(\zeta)) = \overline{\rm Sgn} (\zeta)$ and $\bar b^\psi(\zeta) = \overline{\rm Sgn}(\zeta) B$. 
Then, the glued system is given by
$$\dot {\zeta} = A \zeta+B\overline{\rm Sgn}(\zeta)u \quad \hbox{for $\zeta\in\bar \CC^\psi,$}$$
and becomes a continuous system by \eqref{eq:v}; $u = \overline{\rm Sgn} (\zeta) v$.

Since $(A,B)$ is controllable, we design a state tracking control law:
$$v = v^\psi\cc(v\rr,\zeta\rr,\zeta) = K(\zeta-\zeta\rr) + v\rr$$
where $A+BK$ is Hurwitz.
Then, through \eqref{System: no jump controller}, we finally obtain the local tracking control law for \eqref{Example: linear-2}:
\begin{align*}
u &= u\cc(u\rr,\bar r,\bar x)= \gamma(\bar x) v^\psi\cc(v\rr,\psi(\bar r),\psi(\bar x)) \\
&= K ( x - p\rr(t) r(t) p ) + p\rr(t) u\rr(t) p.
\end{align*}
We illustrate a simulation result in Figure~\ref{Figure: EX:linear system} when
\begin{align*}
A &=\begin{bmatrix}
	0 &1 \\
	0 & 0
\end{bmatrix}, \quad 
B=\begin{bmatrix}
	0\\1
\end{bmatrix}, \quad 
K = [-0.6, ~ -1.55], \\
u\rr(t) &= \left\{\begin{array}{ll} -3 &\text{if } t\;{\rm mod}\;10 \in [0,4), \\ -2&\text{otherwise,}
\end{array} \right .  
\end{align*}
with $r(0)=(0,6)$ and $x_0=(3,8)$.
\hfill $\End$
\end{Exam}

\begin{figure}[t]
\centering
\includegraphics[width=.85\columnwidth]{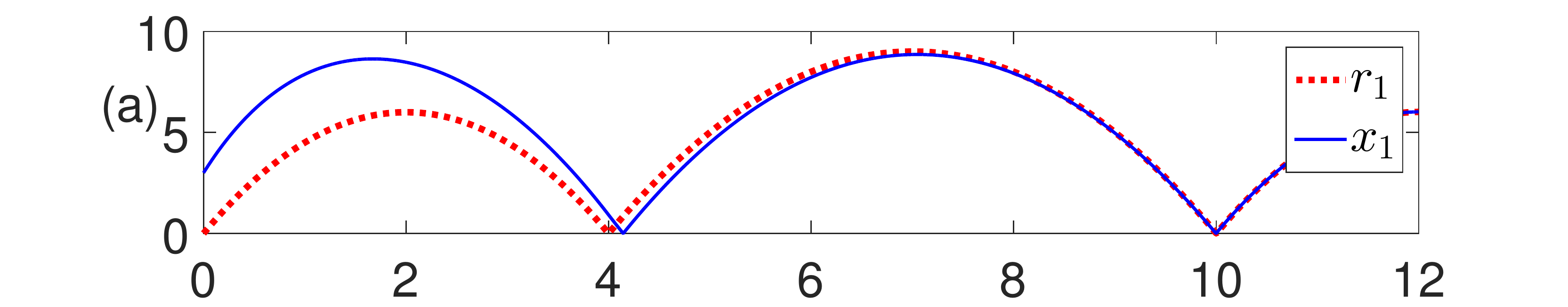}
\includegraphics[width=.85\columnwidth]{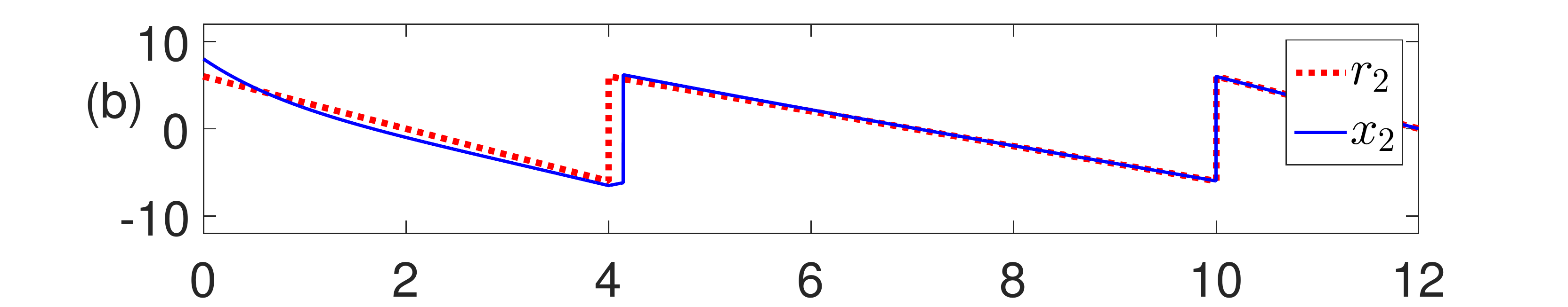}
\includegraphics[width=.85\columnwidth]{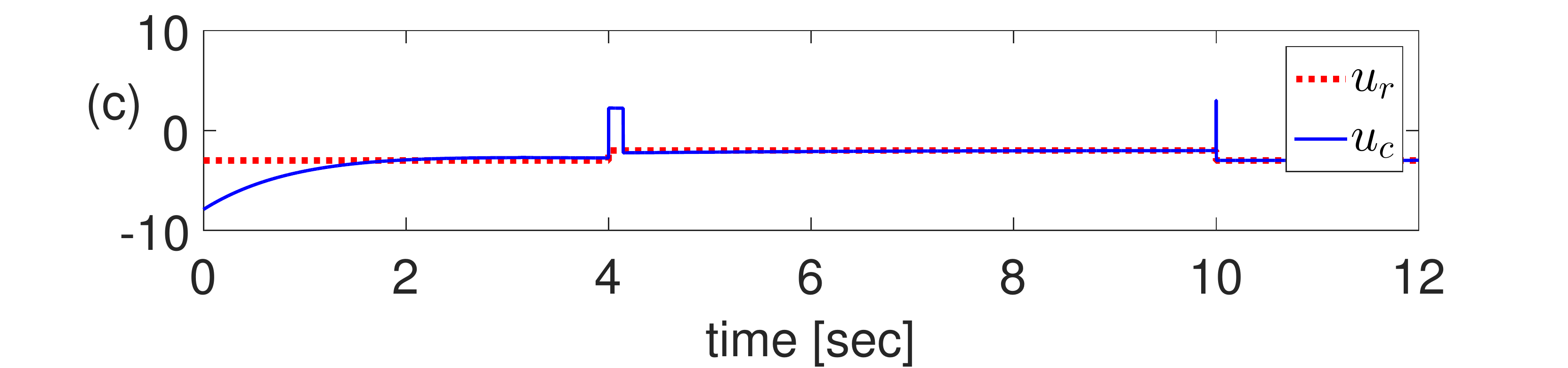}
\caption{Figures (a) and (b) depict the reference $r_1$ and $r_2$ (red) and the plant state $x_1$ and $x_2$ (blue), respectively, and figure (c) depicts the reference input $u\rr$ (red) and the control input $u\cc$ (blue).}
\label{Figure: EX:linear system}
\end{figure}

\section{A Remark}

Even if we presented a new way of constructing state observers and tracking controllers for hybrid systems, there lacks a systematic way to construct the gluing functions satisfying 
\eqref{eq: matching condition of estimation} or \eqref{eq: matching condition for tracking}.
For particular cases, some concrete construction may be developed, as done in Theorem \ref{Thorem: immersion observer}.

\vfill

\pagebreak

\appendix

\noindent{\em Proof of Lemma~\ref{Lem: local embedding of GF}}:
If $\MM\cap \DD=\varnothing$, then $\MM\subset \CC\bs\DD$ and, by (G2), $\psi$ is injective on $\MM$. 
Next we consider the case when $\MM\cap \GG=\varnothing$. 
Suppose that there are $x,y\in \MM\subset \CC\bs \GG$ such that $x\neq y$ and $\psi(x)=\psi(y)$. 
Due to (G2), at least one should be in $\DD$. 
Without loss of generality, we have that $x\in \DD\cap (\CC\bs\GG)$ and $y\in \CC\bs\GG$. 
We first consider the case when $y$ is also included in $\DD$. 
Then, by (G1), it holds that $\psi(g(x))=\psi(g(y))$, which implies that $x=y$ because $g$ is injective on $\DD$ and $\psi$ is injective on $g(\DD)=\GG \subset \CC\bs\DD$. 
This is a contradiction. 
Secondly, consider the case when $y$ is not included in $\DD$. 
Then, it follows from (G1) that $\psi(x)=\psi(x')=\psi(y)$ where $x':=g(x)\in \GG$ and $x'\neq y\in \CC \bs (\DD\cup\GG)$, which is also a contradiction because $x',y\in \CC\bs\DD$ and, by (G2), $\psi$ is injective on $\CC\bs\DD$. 
Therefore, $\psi$ is injective on $\MM$ whenever $\MM\cap \DD=\varnothing$ or $\MM\cap \GG=\varnothing$.

Next, to show
\begin{align*}
	|x-y|\leq L|\psi(x)-\psi(y)|~\hbox{for all $(x,y)\in \MM \times \MM$},
\end{align*} it is enough to check that
$$\inf_{\substack {x\neq y\\ x,y \in \MM} } \frac{|\psi(x)-\psi(y)|}{|x-y|}>0.$$ 
Suppose that there exist sequences $\{x_i\}$ and $\{y_i\}$ in $\MM$ such that $x_i \not = y_i$ and 
\begin{align}\label{eq: lemma2 pf}
\lim_{i\to \infty} \frac{|\psi(x_i)-\psi(y_i)|}{|x_i-y_i|}=0.
\end{align}
Since $\MM\subset \R^n$ is compact, by Bolzano-Weierstrass theorem, without loss of generality, we may assume that $\{x_i\}$ and $\{y_i\}$ converge to some points $x'$ and $y'$ in $\MM$, respectively. 
If $x'\neq y'$, then $\psi(x')=\psi(y')$, which contradicts to injectivity of $\psi$ on $\MM$. 
Therefore, $x'= y'$. 
Since $\psi$ is of class $C^1$, we have
$$\lim_{i\to\infty} \frac{|\psi(x_i)-\psi(y_i)-\dd\psi(x')(x_i-y_i)|}{|x_i-y_i|}=0.$$
Then, from \eqref{eq: lemma2 pf}, we have that 
\begin{align}\label{eq: seq nulling}
\lim_{i\to\infty} |\dd\psi(x')w_i|=0,
\end{align}
where $w_i:=(x_i-y_i)/|x_i-y_i|$ and $|w_i|=1$.
Since $r_\CC$ is smooth and $r_\CC(x)=0_{n-k}$ for all $x\in\CC$, it follows that
\begin{align*}
\lim_{i\to\infty} \frac{|r_\CC(x_i)-r_\CC(y_i)+\dd r_\CC(x')(x_i-y_i)|}{|x_i-y_i|}&=\lim_{i\to\infty}|\dd r_\CC(x')w_i|\\
&=0.
\end{align*}
Take $v_i:=w_i- \dd r_\CC(x')^\top(\dd r_\CC(x')\dd r_\CC(x')^\top)^{-1}\dd r_\CC(x')w_i$. Note that $\dd r_\CC(x')\dd r_\CC(x')^\top$ is invertible because $\dd r_\CC(x)$ has full row rank for every $x\in\CC$ . 
Then, $v_i\in \ker(\dd r_\CC(x'))$ and $|v_i-w_i|\to0$ as $i\to\infty$.
Moreover, since $|v_i|\to1$ as $i\to\infty$, there exists a subsequence $\{v_{i_j}\}$ such that $\lim _{j\to \infty}v_{i_j}=v^* \neq 0_{n}$ for some $v^*\in \ker(\dd r_\CC(x'))$. 
By \eqref{eq: seq nulling}, we obtain that
\begin{align*}
\!|\dd\psi(x')v^*|&=\!\lim_{j\to\infty}\!|\dd\psi(x')v_{i_j}|\\
&\leq\!\lim_{j\to\infty} \!|\dd\psi(x')w_{i_j}\!|\!+\!\!\lim_{j\to\infty}\!|\dd\psi(x')(v_{i_j}\!\!-w_{i_j})|=0,
\end{align*}
which is a contradiction to (G4). 
\hfill $\blacksquare$

\noindent{\em Proof of Lemma~\ref{Lem: Existence of alpha and beta}}:
Define $w(x):=\nabla r_\GG(x)\cdot f(x)$ for $x\in\CC$. Then, by (E2), $w(x)>0$ for all $x\in\GG$. Moreover, $w(\cdot)$ is uniformly continuous on $\EE$, because $r_\GG$ and $f$ are smooth and $\EE$ is compact. Consequently, $\mu:=\inf_{x\in\GG\cap \EE}w(x)$ is positive by the compactness of $\GG\cap \EE$ and there exists $\epsilon_1>0$ such that if $(x_1,x_2)\in\EE\times\EE$ and $|x_1-x_2|<\epsilon_1$, then 
\begin{align*}
|w(x_1)-w(x_2)|<\frac{\mu}{2}.
\end{align*}
If $x\in \OO_{\GG}(\epsilon_1)$, by the definition, there exists $x_\GG\in\GG\cap\EE$ such that $|x-x_\GG|< \epsilon_1$. From the above equation, we have that
\begin{align*}
-\frac{\mu}{2}< w(x)-w(x_\GG)< \frac{\mu}{2}.
\end{align*}   
Then, by the definition of $\mu$, it follows that $w(x)>\mu/2$ for all $x\in \OO_{\GG}(\epsilon_1)$.   

Next, since $r_\GG$ is smooth, it is Lipschitz on the compact set $\EE$ with a Lipschitz constant $L>0$. Moreover, by (A4) and (E2), it is satisfied that $r_\GG(x)>0$ if $x\in \EE\bs\GG$ and  $r_\GG(x)=0$ if $x\in \GG\cap\EE$. From this fact, we obtain that, for all $(x,x_\GG)\in  \EE\times (\GG\cap\EE)$,
\begin{align*}
r_\GG(x)=|r_\GG(x)|=|r_\GG (x)- r_\GG(x_\GG)|&\leq L|x-x_\GG|. 
\end{align*}
Therefore, if $x\in \OO_{\GG}(\epsilon)$, then $r_\GG(x)< L\epsilon$. By contraposition, if $r_\GG(x)\geq  L \epsilon$, then $x\notin \OO_{\GG}(\epsilon)$.  

For a state trajectory $x(t)$ starting on $\EE$, let $\tau$ be its hybrid time trajectory. Then, we have that $x(\tau_i)\in\EE\cap\GG$ and $r_\GG(x(\tau_i))=0$ for $i=1,\dots,N$. Let $\alpha_1(\epsilon):=(2L/\mu)\epsilon$. We first claim that, for all $\epsilon < \epsilon_1$, $x(t)$ escapes from $\OO_{\GG}(\epsilon)$ at least once before $t=\tau_i+\alpha_1({\epsilon})$. Suppose that $x(t)$ remains in $\OO_{\GG}(\epsilon)$ for all $t\in[\tau_i,t_\epsilon]$ where $t_\epsilon\in [\tau_i+\alpha_1(\epsilon),\tau'_{i})$.
Then, since $x(t)\in \OO_{\GG}(\epsilon)\subset\OO_{\GG}(\epsilon_1)$ for $t\in [\tau_i,t_\epsilon]$, it holds that $w(x(t))>\frac{\mu}{2}$ and 
\begin{align*}
r_\GG(x(t_\epsilon))&=r_\GG(x(t_\epsilon))-r_\GG(x(\tau_i))=\int_{\tau_i}^{t_\epsilon}w(x(s))ds\\
&>\frac{\mu}{2}|t_\epsilon-\tau_i|\geq\frac{\mu}{2}\alpha_1(\epsilon)= L\epsilon.
\end{align*}
This implies that $x(t_\epsilon) \notin \OO_{\GG}(\epsilon)$, which is a contradiction. 

Since $r_\GG(x(t))$ increases when $x(t)\in\OO_{\GG}(\epsilon_1)$, $x(t)$ cannot return to $\OO_{\GG}(\epsilon)$ again. 
Therefore, we obtain that $x(t) \notin \OO_\GG(\epsilon)$ for all $t \in [\tau_i + \alpha_1(\epsilon), \tau_i')$ and for all $i = 1,\dots,N$. 
For the case $i=0$, the same claim also holds, without loss of generality, regardless whether $x(\tau_0) \in \GG$ or not.
Therefore, it follows that 
$$x(t) \notin \OO_\GG(\epsilon) \quad \hbox{for all } t \in \bigcup_{i=0}^N[\tau_i+\alpha_1(\epsilon),\tau_i').$$

Similarly, there exist $\epsilon_2>0$ and class-$\KK$ function $\alpha_2$ such that, for $\epsilon<\epsilon_2$, 
$$x(t)\notin\OO_\DD(\epsilon)\hbox{ for all }
t\in\left\{
\begin{array}{ll}
\bigcup_{\,i=0}^{N\!-\!1}\,(\tau_i,\tau_i'-\alpha_2(\epsilon)]\cup(\tau_N,\infty)
\\
\hspace{2.6cm}\hbox{when $N<\infty$,} \\
\bigcup_{i\,=0}^N\,(\tau_i,\tau_i'-\alpha_2(\epsilon)]
\\
\hspace{2.6cm}\hbox{when $N=\infty$.}
\end{array}
\right.$$
Finally, take $\alpha(\cdot) := \max \left (\alpha_1(\cdot),\alpha_2(\cdot)\right )$ and consider $\epsilon<\min(\epsilon_1,\epsilon_2)$. Then, $x(t)\notin \OO_\GG(\epsilon)\cup\OO_\DD(\epsilon)$ for all $t\in\tau_\alpha(\epsilon)$.\hfill $\blacksquare$

\pagebreak
\noindent{\it Proof of Corollary~\ref{Cor:Graphcal sense}}:
If $x(t)$ has a finite number of jumps (i.e. $N<\infty$), then it follows from Theorem \ref{Theorem: State Estimation} for the case where $\epsilon=\epsilon^*$ that (a) and (b) are true by taking $T^*:=\max(\tau_N+\alpha(\epsilon^*), T)$ and $s:=t$. 
Now, we prove the case when $N=\infty$. 

(a): Let $\epsilon:= \min \big(\epsilon^*/4, \alpha^{-1}(\epsilon^*/2), \alpha^{-1}(\epsilon^*/(4M)) \big)$ where $M := \sup_{x \in \EE}{|f(x)|}$. 
If $\epsilon^*$ is sufficiently small, by Theorem \ref{Theorem: State Estimation}, there exists $T>0$ such that $|x(t)-\hat x(t)|<\epsilon$ for $t\in \tau_\alpha(\epsilon)\cap(T,\infty)$. 
Then, there is unique $j\in \langle\tau\rangle$ such that $T\in[\tau_j, \tau_{j}')$. 
We set $T_a:=\tau_{j}'$ and show that (a) holds when $T^*=T_a$. 
With $t>T_a$, three cases are considered; Case 1: $t \in (T_a,\infty)\cap\tau_\alpha(\epsilon)$, Case 2: $t\in [\tau_i, \tau_i+\alpha(\epsilon))$ for $i> j$, and Case 3: $t\in(\tau_i'-\alpha(\epsilon), \tau_i')$ for $i> j$.
In Case 1, take $s=t$. Then, it holds that $|(t,x(t))-(s,\hat x(s))|=|x(t)-\hat x(t)|<\epsilon\leq \epsilon^*/4 < \epsilon^*$. 
In Case 2, take $s=\tau_i+\alpha(\epsilon)$. 
Then, we have that $|s-t|\leq\alpha(\epsilon)\leq \epsilon^*/2$ and
\begin{align*}
&|x(t)-\hat x(s)| \\
&\leq |x(t)-x(\tau_i+\alpha(\epsilon))|+|x(\tau_i+\alpha(\epsilon))-\hat x(\tau_i+\alpha(\epsilon))| \\ 
&< M\alpha(\epsilon)+\epsilon \leq\frac{\epsilon^*}{2}. 
\end{align*}
Therefore, it follows that $|(t,x(t))-(s,\hat x(s))|<\epsilon^*.$ Similarly, in Case 3, it is satisfied that $|(t,x(t))-(s,\hat x(s))|<\epsilon^*$ when we take $s=\tau_i-\alpha(\epsilon)$.

(b): Consider the sufficiently small $\epsilon^*$ satisfying the condition for $\epsilon$ in Theorem~\ref{Theorem: State Estimation}, (i.e., $\OO_\DD(\epsilon^*) \cap \OO_\GG(\epsilon^*) = \varnothing$). 
As shown in the proof of Theorem~\ref{Theorem: State Estimation}, by Lemma~\ref{Lem: local embedding of GF}, we can take $L_1(\epsilon^*)>0$ and $L_2(\epsilon^*)>0$ such that $|x-x'|\leq L_1|\psi(x)-\psi(x')|$ for all $x,x'\in \EE\bs\OO_\DD(\epsilon^*)$ and $|x-x'|\leq L_2|\psi(x)-\psi(x')|$ for all $x,x'\in \EE\bs\OO_\GG(\epsilon^*)$. 

Let $\epsilon := \min\big (\epsilon^*, \alpha^{-1}(\epsilon^*/4), \alpha^{-1}(\epsilon^*/(4LM^\psi)), \alpha^{-1}(\epsilon^*/$ $(2M)) \big)$ where $M^\psi := \max_{x \in \EE}{|\dd \psi(x) f(x)|}$ and $L := 2\cdot\max(L_1,L_2)$. 
Then, by Theorem \ref{Theorem: State Estimation}, there exist $T_\zeta>0$ and $T>0$ such that $|\zeta(t)-\Pi_{\psi(\EE)}(\hat\zeta(t))| < 2\gamma = \epsilon^*/(2L)$ for $t>T_\zeta$ and $|x(t)-\hat x(t)|<\epsilon$ for $t\in \tau_\alpha(\epsilon)\cap(T,\infty)$. 
There exists unique $j\in \langle\tau\rangle$ such that $\max(T_\zeta, T)\in[\tau_j, \tau_{j}')$. 
We set $T_b := \tau_{j}'+\alpha(\epsilon)$ and show that (b) holds when $T^*=T_b$. 
Again, with $t>T_b$, we consider three cases; Case 1: $t \in  (T_b,\infty)\cap\tau_\alpha(\epsilon)$, Case 2: $t \in (T_b,\infty)\bs \tau_\alpha(\epsilon)$ and $\hat x(t) \in \EE\bs\OO_\DD(\epsilon^*)$, and Case 3: $t \in (T_b,\infty)\bs \tau_\alpha(\epsilon)$ and $\hat x(t) \in \EE\bs\OO_\GG(\epsilon^*)$.
In Case 1, take $s=t$. 
Then, $|(s,x(s))-(t,\hat x(t))|=|x(t)-\hat x(t)|<\epsilon\leq\epsilon^*$. 
In Case 2, let $s=\tau_{i^*}+\alpha(\epsilon)$ where $i^*$ is the positive integer (larger than $j+1$) such that $t \in (\tau_{i^*} - \alpha(\epsilon), \tau_{i^*} + \alpha(\epsilon))$. 
Since $|s-t| < 2\alpha(\epsilon)\leq \epsilon^*/2$, we just show that $|x(s)-\hat x(t)| \leq \epsilon^*/2$ for the proof of this case. 
Since $t>T_\zeta$, it is satisfied that
\begin{align*}
|\zeta(s)-\Pi_{\psi(\EE)}(\hat\zeta(t))| &\leq |\zeta(\tau_{i^*}\!+\!\alpha(\epsilon))\!-\!\zeta(t)|\! \\
&\hspace{-2cm}+\!|\zeta(t)\!-\!\Pi_{\psi(\EE)}\!(\hat\zeta(t))|< 2M^\psi\alpha(\epsilon)+\frac{\epsilon^*}{2L}\leq \frac{\epsilon^*}{L}.
\end{align*}
In addition, since $|x(s)-x(\tau_{i^*})| = |x(\tau_{i^*}+\alpha(\epsilon)) - x(\tau_{i^*})| \leq M \alpha(\epsilon) \leq \epsilon^*/2$ and $x(\tau_{i^*})\in\GG$, we have that $x(s)\in\OO_\GG(\epsilon^*)\subset \EE\bs\OO_\DD(\epsilon^*)$. 
Therefore, it follows that
\begin{align*}
|x(s)-\hat x(t)| &\leq L_1 |\psi(x(s))-\psi(\hat x(t))| \\
&= L_1 |\zeta(s)-\Pi_{\psi(\EE)}(\hat\zeta(t))|<\frac{\epsilon^*}{2}.
\end{align*}
In a similar way, we can prove Case 3 by setting $s := \tau_{i^*} - \alpha(\epsilon)$. 
Finally, let $T^* := \max(T_a, T_b)$, which completes the proof. \hfill $\blacksquare$

\medskip

\noindent{\it Proof of Theorem~\ref{Theorem: Lip conti}}:
By the vector field matching condition \eqref{eq: matching condition of estimation}, it is easy to show the continuity of $f^\psi$. We now show that $f^\psi$ is locally Lipschitz at every $\zeta\in\CC^\psi$. 

For any $\zeta\in\CC^\psi$, it follows from (G1) and (G2) that either $\zeta\in\psi(\CC\bs(\DD\cup\GG))$ or $\zeta\in\psi(\DD\cup\GG)=\psi(\DD)$. The former means $\zeta$ is a image of an interior point of $\CC$ while the latter means $\zeta$ is a image of a boundary point of $\CC$.

Suppose that $\zeta\in\psi(\CC\bs(\DD\cup\GG))$. Then, there is $x\in \CC\bs(\DD\cup\GG)$ such that $\psi(x)=\zeta$ where $x$ is an interior point. By the inverse function theorem, there exists an open neighborhood $U$ of $x$ in $\CC$ such that $U\cap \partial \CC =\varnothing$, $\psi|_U$ is injective, and $\psi|_U^{-1}$ is continuously differentiable. Note that $\psi(U)$ is open in $\R^m=\R^k$. Therefore, there exist $L(\zeta)>0$ and open neighborhood $V$ of $\zeta$ in $\psi(U)$ such that $|\psi|_U^{-1}(\zeta_1)-\psi|_U^{-1}(\zeta_2)|\leq L(\zeta)|\zeta_1-\zeta_2|$ for $\zeta_1,\zeta_2\in  V$. Since $\dd\psi$ and $f$ are continuously differentiable, there also exists $L_0(\zeta)>0$ such that
\begin{align*}
&|f^\psi(\zeta_1)-f^\psi(\zeta_2)| \\
&= |\dd\psi(\psi|_U^{-1}(\zeta_1)) f(\psi|_U^{-1}(\zeta_1))-\dd\psi(\psi|_U^{-1}(\zeta_2)) f(\psi|_U^{-1}(\zeta_2))| \\
&\leq L_0(\zeta)|\psi|_U^{-1}(\zeta_1) - \psi|_U^{-1}(\zeta_2)| \leq L_0(\zeta)L(\zeta)|\zeta_1-\zeta_2|,
\end{align*}
for any $\zeta_1,\zeta_2\in V\subset \CC^\psi$. 
Consequently, $f^\psi$ is locally Lipschitz at every $\zeta\in\psi(\CC\bs(\DD\cup\GG))$ on $\CC^\psi$. 

Next we consider the case when $\zeta\in\psi(\DD)$. 
Then, there is $(x_\DD, x_\GG)\in \DD \times \GG$ such that $x_\GG=g(x_\DD)$ and $\psi(x_\DD)=\psi(x_\GG)=\zeta$. 
By the inverse function theorem, there exist open neighborhoods $U_\DD$ of $x_\DD$ and $U_\GG$ of $x_\GG$ such that $U_\DD\times U_\GG \subset  \CC \times  \CC$ and $\psi_1^{-1}$ and $\psi_2^{-1}$ are continuously differentiable where $\psi_1:=\psi|_{U_\DD}$ and $\psi_2:= \psi|_{U_\GG}$. In addition, since $\psi_1^{-1}$ is continuously differentiable, there exists a Lipschitz constant $L_1(\zeta)>0$ on an open neighborhood $V_\DD$ of $\zeta$ in $\psi(U_\DD)$ and it follows that, for $\zeta_1,\zeta_2\in V_\DD$,
\begin{align*}
&|f^\psi(\zeta_1)-f^\psi(\zeta_2)|= \notag\\
&|\dd\psi(\psi_1^{-1}(\zeta_1)) f(\psi_1^{-1}(\zeta_1))-\dd\psi(\psi_1^{-1}(\zeta_2)) f(\psi_1^{-1}(\zeta_2))|
\\
&~~~~~~\leq L_0(\zeta)|\psi_1^{-1}(\zeta_1) - \psi_1^{-1}(\zeta_2)| \leq L_0(\zeta)L_1(\zeta)|\zeta_1-\zeta_2|.\notag
\end{align*}
Similarly, for the second case, we can take $L_2(\zeta)>0$ on an open neighborhood $V_\GG$ of $\zeta$ in $\psi(U_\GG)$ such that, for $\zeta_1,\zeta_2\in V_\GG$,
\begin{align*}
|f^\psi(\zeta_1)-f^\psi(\zeta_2)|\leq L_0(\zeta)L_2(\zeta)|\zeta_1-\zeta_2|.
\end{align*}
For sufficiently small $\delta>0$ such that $\BB_{\zeta}(\delta)\subset V_\DD \cup V_\GG$, if $\zeta_1,\zeta_2\in\BB_{\zeta}(\delta)$, without loss of generality, one of the following three cases holds:
\begin{itemize}
	\item $\psi^{-1}(\zeta_1)=\psi_1^{-1}(\zeta_1)\in V_\DD$ and $\psi^{-1}(\zeta_2)=\psi_1^{-1}(\zeta_2)\in V_\DD$.
	\item $\psi^{-1}(\zeta_1)=\psi_2^{-1}(\zeta_1)\in V_\GG$ and $\psi^{-1}(\zeta_2)=\psi_2^{-1}(\zeta_2)\in V_\GG$.
	\item $\psi^{-1}(\zeta_1)=\psi_1^{-1}(\zeta_1)\in V_\DD$ and $\psi^{-1}(\zeta_2)=\psi_2^{-1}(\zeta_2)\in V_\GG$.
\end{itemize}
For the first and second cases, we obtain that
\begin{align} \label{eq. Lipschitz Condition 1}
|f^\psi(\zeta_1)-f^\psi(\zeta_2)|\leq L_0(\zeta)\cdot\max(L_1(\zeta),L_2(\zeta))|\zeta_1-\zeta_2|.
\end{align}
Let us consider the last case. 
Since a sufficiently small neighborhood of $\zeta$ in $\CC^\psi$ is divided into two regions by $\psi(\DD)$ and the last case implies that $\zeta_1$ and $\zeta_2$ are placed at the different regions, there exists at least one $\zeta^*\in \psi(\DD)\cap l(\zeta_1,\zeta_2)$, where $l(\zeta_1,\zeta_2)$ is the line segment whose end points are $\zeta_1$ and $\zeta_2$. 
Then, we find $ x^*\in\DD$ such that $\psi_1( x^*)=\psi_2( g( x^*))=\zeta^*$ by (G1). 
Thus, it follows from \eqref{eq: matching condition of estimation} that
\begin{align}
&|f^\psi(\zeta_1)-f^\psi(\zeta_2)|\notag\\
&\leq |f^\psi(\zeta_1)\!-\!\dd \psi( x^*) f( x^*)|
+ |\dd\psi( g( x^*)) f( g( x^*))\!-\!f^\psi(\zeta_2) |\notag \\
&= |\dd\psi(\psi_1^{-1}(\zeta_1)) f (\psi_1^{-1}(\zeta_1))-\dd\psi(\psi_1^{-1}(\zeta^*))  f(\psi_1^{-1}(\zeta^*))|\notag\\
&\quad+ |\dd\psi(\psi_2^{-1}(\zeta^*))  f(\psi_2^{-1}(\zeta^*))-\dd\psi(\psi_2^{-1}(\zeta_2)) f(\psi_2^{-1}(\zeta_2))|\notag\\
&\leq L_0(\zeta)\left (|\psi_1^{-1}(\zeta_1)-\psi_1^{-1}(\zeta^*)|+|\psi_2^{-1}(\zeta^*)-\psi_2^{-1}(\zeta_2)|\right)\notag\\
&\leq L_0(\zeta)\left (L_1(\zeta)|\zeta_1-\zeta^*| + L_2(\zeta)|\zeta^*-\zeta_2|\right )\notag\\
&\leq L_0(\zeta)\cdot\max(L_1(\zeta),L_2(\zeta))|\zeta_1-\zeta_2|. \label{eq. Lipschitz Condition 3}
\end{align}
Consequently, by \eqref{eq. Lipschitz Condition 1}--\eqref{eq. Lipschitz Condition 3}, $f^\psi$ is locally Lipschitz at every $\zeta\in\psi(\DD)$ on $\CC^\psi$ with the Lipschitz constant $L_0(\zeta)\cdot\max(L_1(\zeta),L_2(\zeta))$.

Since $\EE$ is compact and $\psi$ is continuous, $\psi(\EE)$ is also compact. 
Therefore, it follows that $f^\psi$ is Lipschitz continuous on $\psi(\EE)$.
In the similar way, we can show that $h^\psi$ is also Lipschitz continuous on $\psi(\EE)$.\hfill $\blacksquare$

\end{document}